\documentclass[a4paper,12pt]{elsarticle}
\makeatletter
\def\ps@pprintTitle{%
 \let\@oddhead\@empty
 \let\@evenhead\@empty
 \def\@oddfoot{\centerline{\thepage}}%
 \let\@evenfoot\@oddfoot}
\makeatother
\usepackage{amsmath,amsthm}
\usepackage{mathtools}
\usepackage[linesnumbered,ruled,vlined]{algorithm2e}
\usepackage{algpseudocode}
\newtheorem{prop}{Proposition}
\newtheorem{Lemma}{Lemma}
\newtheorem{Remark}{Remark}

\begin{document}
\begin{frontmatter}

%\begin{document}

%\newtheorem{Proposition}
% Le titre du papier
\title{No-idle, no-wait: when shop scheduling meets dominoes, eulerian and hamiltonian paths}

% Le titre court
\def\shorttitle{Shop schedules, dominoes and paths}

% Les auteurs et leur num�ro d'affiliation
\author[1]{J-C. Billaut}
\author[2]{F. Della Croce}
\author[2]{F. Salassa}
\author[1]{V. T'kindt}

\address[1]{\small Universit\'{e} Francois-Rabelais de Tours, ERL CNRS OC 6305, Tours, France}
\address[2]{\small DIGEP, Politecnico di Torino, Corso Duca degli Abruzzi 24, Torino, Italy}

\begin{abstract}
In shop scheduling, several applications exist where it is required that some components 
perform consecutively. We refer to no-idle schedules if  
machines are required to operate with no inserted idle time and no-wait schedules if 
tasks cannot wait between the end of an operation and the start of the following one.
We consider here no-idle/no-wait shop scheduling problems with makespan
as performance measure and determine related complexity results.
We first analyze 
%problem $F2 | no-idle, no-wait | C_{\max}$ 
the two-machine no-idle/no-wait flow shop problem and 
show that it is equivalent to a special version of the game of dominoes which is
polynomially solvable by tackling an Eulerian path problem on a directed graph. 
We present for this problem an $O(n)$ exact algorithm.
As a byproduct we show that 
%We show that 
%problem $F2 | no-idle, no-wait | C_{\max}$ 
the Hamiltonian Path problem on a digraph  
$G(V,A)$ with a special structure (where every pair of vertices $i,j$
either has all successors in common or has no common successors)
reduces to the two-machine no-idle/no-wait flow shop problem.
%(in other words, either the successors of $i$ coincide with the successors of $j$ or $i,j$ have 
%distinct successors).
Correspondingly, we provide a new polynomially solvable special case of the Hamiltonian Path problem.
%We denote this latter problem as the Common/Distinct Successors Hamiltonian cycle (CDSHDP) problem. 
Then, we show that also the corresponding $m$-machine no-idle no-wait flow shop problem %$F | no-idle, no-wait | C_{\max}$ 
is polynomially solvable and provide an $O(mn\log n)$  exact algorithm. Finally we prove that the 
2-machine no-idle/no-wait job shop problem
and the 2-machine no-idle/no-wait open shop problem
%problems $J2 | no-idle, no-wait | C_{\max}$ 
%and $O2 | no-idle, no-wait | C_{\max}$
are $NP$-Hard in the strong sense.
\end{abstract}


\end{frontmatter}

\section{Introduction}

In shop scheduling, typically when machines represent very expensive equipments and the fee is directly linked to the actual time consumption, it is of interest to 
determine so-called no-idle solutions, that is schedules where machines process the jobs continuously without inserted idle time.
Also, particularly in metal-processing industries, where
delays between operations interfere with the technological process,
it is required to obtain no-wait schedules where each job is subject to the so-called no-wait constraint, that is it cannot be idle between the completion of an operation and the start of the following one. 
\bigskip

\noindent

We attack here the simultaneous combination of these two requirements by considering no-idle/no-wait shop scheduling problems with three different shop configurations namely flow shop, job shop and open shop.
We focus on the makespan as performance measure.  
We first deal with the two-machine flow shop problem and the m-machine flow shop problem showing that both are polynomially solvable and present connections with the game of dominoes and well-known graph problems. 
Then, we prove that the two-machine job shop and the two-machine open shop are $NP$-hard in the strong sense.
Using the standard three-field notation \cite{LLRS93}, 
the three two-machine shop problems are denoted as 
$F2 | no-idle, no-wait | C_{\max}$ for the flow shop, $J2 | no-idle, no-wait | C_{\max}$ for the job shop and  
$O2 | no-idle, no-wait | C_{\max}$ for the open shop, respectively, while the general flow shop case
is denoted as $F | no-idle, no-wait | C_{\max}$.

\bigskip

\noindent
With respect to the relevant literature, in one of the pioneering works in scheduling \cite{Johnson54}, it is shown that problem 
$F2 || C_{\max}$ is solvable in $O(n log n)$ time by first arranging the jobs with $p_{1,j}\leq p_{2,j}$ in non-decreasing order 
of $p_{1,j}$, followed by the remaining jobs
arranged in non-increasing order of $p_{2,j}$, where $p_{i,j}$ denotes the processing time of job $J_j$  on machine $M_i$.
As mentioned in \cite{Adiri}, $F2 | no-idle | C_{\max}$ can also be solved in $O(n log n)$ time by simply packing the jobs on the second machine once the schedule computed by the algorithm in \cite{Johnson54} is given.
In \cite{Reddi72}, it is shown that problem $F2 | no-wait | C_{\max}$ can be seen as a special case of the Gilmore-Gomory
Travelling Salesman Problem \cite{gigo64} and therefore is solvable too in $O(n log n )$ time.
Besides, problems $F3 || C_{\max}$, $F3 | no-idle | C_{\max}$ and $F3 | no-wait | C_{\max}$
were all shown to be $NP$-hard in the strong sense by \cite{gjs76}, \cite{bl97} and \cite{Rock84} respectively.
In \cite{Adiri}, it is reported that 
both problems $F2 | no-idle| \sum C_j$
and  $F2 | no-wait| \sum C_j$ are $NP$-hard by exploiting the fact that the $NP$-hardness proof 
of problem $F2 | | \sum C_j$ in \cite{gjs76} was given by constructing a flow shop instance 
that happened to be both no-idle and no-wait.
Similar consideration holds for problem $F2 | no-idle, no-wait| \sum C_j$.
For surveys on $no-idle$ flow shop scheduling, we refer to
\cite{gs09,CH14}.
For surveys on no-wait scheduling, we refer to \cite{HS96,AL16}.
In \cite{kk07}, the links between problems $F | no-idle | C_{\max}$ and $F2 | no-wait | C_{\max}$
are discussed and some efficiently solvable special cases are shown.
The recent literature on $no-wait$ flow shop scheduling includes \cite{hjm12} where it is shown that
minimizing the number of interruptions on the last machine is solvable in $O(n^2)$ time on two machines (the problem is denoted as 
$F2 | no-wait| \cal{G}$) while it is 
$NP$-hard on three or more machines. 
Finally, we mention the contribution of \cite{GI01}, where it is shown that, if some processing times are allowed  to be zero and zero processing times imply that the corresponding operations
should not be performed, then problem $F2 | no-idle, no-wait | C_{\max}$ is 
$NP$-hard in the strong sense.
Here, we deal with the standard versions of 
no-idle/no-wait shop problems where all processing times are required to be strictly positive.
\bigskip

\noindent

%\noindent
%We consider two standard different shop configurations namely flow shop 
%(the order of processing is $M_1 \rightarrow M_2$ for all jobs) and job shop
%(there are two job subsets and each job is restricted to having exactly two operations; the jobs in one job subset follow the 
%$M_1 \rightarrow M_2$  processing route while the jobs in the other job subset follow the  $M_2 \rightarrow M_1$ processing route).
%\noindent

%We will focus primarily on the makespan as performance measure. Using the general three-field notation, 
%the shop problems are denoted as 
%$F2 | no-idle, no-wait | C_{\max}$ for the flow shop and $J2 | no-idle, no-wait | C_{\max}$ for the job shop , respectively.

%\bigskip

\bigskip

The paper proceeds as follows.
In Section \ref{fss}, we show that 
problem $F2 | no-idle, no-wait | C_{\max}$ is equivalent to an oriented version of the 
Single Player Dominoes problem which has been shown in \cite{dmw14}
to be polynomially solvable and 
presents an $O(n)$ time solution approach.
As a byproduct, we also consider a special case of the Hamiltonian Path problem (denoted as Common/Distinct Successors Directed Hamiltonian Path - CDSDHP - problem) on a directed graph $G(V,A)$ with a specific structure
so that every pair of vertices $i,j$
either has all successors in common or has no common successor.
In other words, either the successors of $i$ coincide with the successors of $j$ or $i,j$ have 
distinct successors. We prove that problem CDSDHP reduces to problem $F2 | no-idle, no-wait | C_{\max}$.
Correspondingly, we provide a new polynomially solvable special case of the Hamiltonian Path problem.
In Section \ref{fmss}, we show that also
the general $F | no-idle, no-wait | C_{\max}$ problem with $m$ machines is polynomially solvable and
present an $O(mn \log n)$ time solution approach.
Finally, section \ref{jss} provides the unary $NP$-Hardness proof of problems $J2 | no-idle, no-wait | C_{\max}$ and 
$O2 | no-idle, no-wait | C_{\max}$.

\section{Two-machine no-idle no-wait flow shop scheduling}
\label{fss}

In the $F2 | no-idle, no-wait | C_{\max}$ problem,  a set of $n$ jobs is available at time zero. Each job $j$ must be processed  non-preemptively on two continuously available machines $M_1,M_2$ with known integer processing times  $p_{1,j}, p_{2,j} > 0$, respectively. 
Each machine is subject to the so-called no-idle constraint, namely, it processes continuously one job at a time, and operations of each job cannot overlap. Each job is subject to the so-called no-wait constraint, namely, it cannot be idle between the completion of the first operation and the start of the second operation.
All jobs are processed first on machine $M_1$ and next on machine $M_2$ and, given the no-wait constraint, the jobs sequences on the two machines must be identical. 
Let us denote by $p(A) = \sum_{j=1}^n p_{1,j}$ the sum of processing times on the first machine and by $p(B) = \sum_{j=1}^n p_{2,j}$ the sum of processing times on the second machine.
%Each machine can process at most one job at a time and the operations of each job cannot %overlap. Also, f
For any given sequence
$\sigma$, $[j]_{\sigma}$ denotes the job in position $j$. 
%Notice, that the extension to the general m-machine case implies that each job $j$ has processing times $p_{i,j}$, for all machines $M_1, ..., M_m$.
\bigskip

\noindent
Consider Figure \ref{f2_picturea} which provides an illustrative example of a feasible no-idle, no-wait schedule for a 4-job problem.

\begin{figure*}[!ht]
    \begin{center}
        \ifx\JPicScale\undefined\def\JPicScale{1}\fi
\unitlength \JPicScale mm
\begin{picture}(80,18)(0,0)
\linethickness{0.3mm}
\put(7,2){\line(1,0){73}}
\linethickness{0.3mm}
\put(7,0){\line(0,1){3}}
\linethickness{0.3mm}
\put(17,0){\line(0,1){3}}
\linethickness{0.3mm}
\put(27,0){\line(0,1){2}}
\linethickness{0.3mm}
\put(37,0){\line(0,1){2}}
\linethickness{0.3mm}
\put(47,0){\line(0,1){2}}
\linethickness{0.3mm}
\put(57,0){\line(0,1){2}}
\linethickness{0.3mm}
\put(67,0){\line(0,1){2}}
\linethickness{0.3mm}
\put(77,0){\line(0,1){3}}
\linethickness{0.3mm}
\put(12,0){\line(0,1){3}}
\linethickness{0.3mm}
\put(22,0){\line(0,1){2}}
\linethickness{0.3mm}
\put(32,0){\line(0,1){2}}
\linethickness{0.3mm}
\put(42,0){\line(0,1){2}}
\linethickness{0.3mm}
\put(52,0){\line(0,1){2}}
\linethickness{0.3mm}
\put(62,0){\line(0,1){2}}
\linethickness{0.3mm}
\put(72,0){\line(0,1){3}}
\linethickness{0.3mm}
\put(7,12){\line(1,0){73}}
\linethickness{0.3mm}
\put(7,10){\line(0,1){3}}
\linethickness{0.3mm}
\put(17,10){\line(0,1){2}}
\linethickness{0.3mm}
\put(27,10){\line(0,1){2}}
\linethickness{0.3mm}
\put(37,10){\line(0,1){2}}
\linethickness{0.3mm}
\put(47,10){\line(0,1){2}}
\linethickness{0.3mm}
\put(57,10){\line(0,1){2}}
\linethickness{0.3mm}
\put(67,10){\line(0,1){3}}
\linethickness{0.3mm}
\put(77,10){\line(0,1){3}}
\linethickness{0.3mm}
\put(12,10){\line(0,1){2}}
\linethickness{0.3mm}
\put(22,10){\line(0,1){2}}
\linethickness{0.3mm}
\put(32,10){\line(0,1){2}}
\linethickness{0.3mm}
\put(42,10){\line(0,1){2}}
\linethickness{0.3mm}
\put(52,10){\line(0,1){2}}
\linethickness{0.3mm}
\put(62,10){\line(0,1){3}}
\linethickness{0.3mm}
\put(72,10){\line(0,1){3}}
\put(-1,15){\makebox(0,0)[cc]{$M_1$}}

\put(-1,4){\makebox(0,0)[cc]{$M_2$}}

\linethickness{0.3mm}
\put(7,12){\line(1,0){12}}
\put(19,12){\line(0,1){6}}
\put(7,12){\line(0,1){6}}
\put(7,18){\line(1,0){12}}
\linethickness{0.3mm}
\put(19,2){\line(1,0){25}}
\put(19,2){\line(0,1){6}}
\put(44,2){\line(0,1){6}}
\put(19,8){\line(1,0){25}}
\linethickness{0.3mm}
\put(19,12){\line(1,0){25}}
\put(19,12){\line(0,1){6}}
\put(44,12){\line(0,1){6}}
\put(19,18){\line(1,0){25}}
\linethickness{0.3mm}
\put(44,8){\line(1,0){10}}
\put(44,2){\line(0,1){6}}
\put(54,2){\line(0,1){6}}
\put(44,2){\line(1,0){10}}
\linethickness{0.3mm}
\put(44,18){\line(1,0){10}}
\put(44,12){\line(0,1){6}}
\put(54,12){\line(0,1){6}}
\put(44,12){\line(1,0){10}}
\linethickness{0.3mm}
\put(54,2){\line(1,0){8}}
\put(54,2){\line(0,1){6}}
\put(62,2){\line(0,1){6}}
\put(54,8){\line(1,0){8}}
\linethickness{0.3mm}
\put(54,18){\line(1,0){8}}
\put(54,12){\line(0,1){6}}
\put(62,12){\line(0,1){6}}
\put(54,12){\line(1,0){8}}
\linethickness{0.3mm}
\put(62,8){\line(1,0){11}}
\put(62,2){\line(0,1){6}}
\put(73,2){\line(0,1){6}}
\put(62,2){\line(1,0){11}}
\put(13,15){\makebox(0,0)[cc]{$P_{1,[1]}$}}

\put(31,15){\makebox(0,0)[cc]{$P_{1,[2]}$}}

\put(31,5){\makebox(0,0)[cc]{$P_{2,[1]}$}}

\put(49,15){\makebox(0,0)[cc]{$P_{1,[3]}$}}

\put(49,5){\makebox(0,0)[cc]{$P_{2,[2]}$}}

\put(58,15){\makebox(0,0)[cc]{$P_{1,[4]}$}}

\put(68,5){\makebox(0,0)[cc]{$P_{2,[4]}$}}

\put(58,5){\makebox(0,0)[cc]{$P_{2,[3]}$}}

\end{picture}				
        \caption[Fig]{A no-idle no-wait schedule for a 2-machine flow shop}
        \label{f2_picturea}
    \end{center}
\end{figure*}
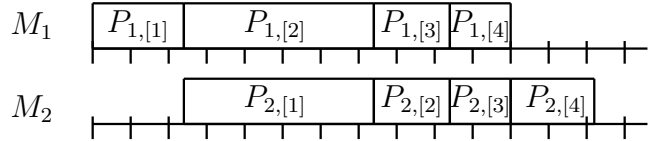

We point out  that the $no-idle,no-wait$ requirement is very strong.
Indeed, any feasible sequence $\sigma$ for problem $F2 | no-idle,no-wait | C_{\max}$,
forces consecutive jobs to share common processing times in such a way that
\begin{eqnarray}\label{myeq}
p_{2,[j]_{\sigma}} = p_{1,[j+1]_{\sigma}} \;\; \forall j \in \,...,n-1.   
\end{eqnarray}

As mentioned in the introduction, a related problem denoted by $F2 | no-wait| \cal{G}$ is tackled in \cite{kk07}, where the aim is to minimize the number of interruptions (idle times) on $M_2$. We remark that this problem does not constitute a generalization 
of problem $F2 | no-idle, no-wait | C_{\max}$ since an optimal solution with no interruptions 
of problem $F2 | no-wait| \cal{G}$ may be non-optimal for problem $F2 | no-idle, no-wait | C_{\max}$.
Consider a $2-job$ instance with processing times $p_{1,1}=b$, $p_{2,1}=a$, $p_{1,2}=a$, $p_{2,2}=b$, with $b>a$.
Then, sequence $(J_1,J_2)$ is no-idle, no-wait, has makespan $C_{\max}^{(J_1,J_2)}=2b+a$ and is optimal for problem $F2 | no-wait| \cal{G}$ 
as it has no interruptions. However, it is not optimal for problem
$F2 | no-idle, no-wait | C_{\max}$ as sequence $(J_2,J_1)$ is also no-idle, no-wait and has makespan $C_{\max}^{(J_2,J_1)}=2a+b < 2b+a$.
%To the authors knowledge, no works are available on shop %scheduling dealing with the contemporaneous presence of no-idle %and no-wait
%constraints. In this section we focus on the complexity status %of the two-machine no-idle no-wait flow shop problem.\\

\subsection{The $F2 | no-idle, no-wait | C_{\max}$ problem and the game of dominoes}
We first provide a lemma on specific conditions of any instance of problem $F2 | no-idle, no-wait | C_{\max}$
for which a feasible solution may exist.

\begin{Lemma}\label{l0}
\noindent
\begin{enumerate}
\item[(C1)] A necessary condition to have a feasible solution for problem $F2 | no-idle, no-wait | C_{\max}$ is 
that there always exists an indexing of the jobs so that $p_{1,2},...p_{1,n}$ and $p_{2,1},...,p_{2,n-1}$
constitute different permutations of the same vector of elements.
\item[(C2)] When the above condition (C1) holds, then
\begin{enumerate}
\item[Case 1] if $p_{1,1} \neq p_{2,n}$, every feasible sequence must have a job with processing time $p_{1,1}$ in first position and a job 
with processing time $p_{2,n}$ in last position.
%Correspondingly the makespan is $p_{1,1}+p(B)$.
%\item[Case 2] if $p_{1,1} = p_{2,n}$ and there exists a feasible sequence,
\item[Case 2] if $p_{1,1} = p_{2,n}$ 
then there exists at least $n$ feasible sequences each starting with a different job by simply rotating the starting sequence
as in a cycle.
\end{enumerate}
\end{enumerate}
\end{Lemma}
\begin{proof}
Condition (C1) trivially holds from expression (\ref{myeq}).\\
For condition C2, if $p_{1,1} \neq p_{2,n}$, then the processing time on the first machine of the job in first position must be equal to $p_{1,1}$ (similarly the processing time on the second machine of the job in last position must be equal to $p_{2,n}$) or else 
there is no way to fulfil expression (\ref{myeq}). 
Besides, If $p_{1,1} = p_{2,n}$ and a feasible sequence $\sigma$ exists fulfilling expression (\ref{myeq}), then
$p_{1,[1]_{\sigma}} = p_{2,[n]_{\sigma}}$ also holds. But then, any forward or backward rotation of $\sigma$ also provides a feasible solution. 
\end{proof}

Next, we show that for any feasible solution the makespan is determined by
the processing time of the first job on the first machine plus $p(B)$.

The following lemma holds.

\begin{Lemma}\label{l1}
The makespan of any feasible sequence $\sigma$ is given by the processing time 
of the first job on the first machine plus the sum of jobs processing times on the second machine.
\end{Lemma}

\begin{proof}
As mentioned above, for any feasible sequence $\sigma$, we have 
$p_{2,[j]_{\sigma}} = p_{1,[j+1]_{\sigma}}$ $\forall j \in 1,...,n-1$.
Correspondingly, 
\begin{eqnarray}
C_{[1]_{\sigma}}=p_{1,[1]_{\sigma}}+p_{2,[1]_{\sigma}} \nonumber \\
C_{[2]_{\sigma}}=p_{1,[1]_{\sigma}}+p_{2,[1]_{\sigma}}+p_{2,[2]_{\sigma}} \nonumber \\
\cdots \nonumber \\
C_{[j]_{\sigma}}=p_{1,[1]_{\sigma}}+\sum_{i=1}^jp_{2,[i]_{\sigma}}  \nonumber 
\end{eqnarray}\label{eq1}

Hence,
\begin{eqnarray}
C_{\max_{\sigma}} = C_{[n]_{\sigma}} =p_{1,[1]_{\sigma}}+\sum_{i=1}^np_{2,[i]_{\sigma}} = p_{1,[1]_{\sigma}}+\sum_{i=1}^np_{2,i}.
% = a_{[1]_{\sigma}}+\sum_{i=1}^n a_{[i]}.
\end{eqnarray}\label{eq2}
\end{proof}

\bigskip
\bigskip

\noindent
%Due to Lemma \ref{l1}, if a feasible solution exists, then 
One may consider approaching problem $F2 | no-idle,no-wait | C_{\max}$
by solving first problem $F2 | no-idle,no-wait | \cal{G}$. Then, if 
a solution without interruption is found, it is immediate to determine 
the optimal $F2 | no-idle,no-wait | C_{\max}$ sequence by exploiting 
Lemmata \ref{l0}, \ref{l1}. Alternatively, $F2 | no-idle,no-wait | C_{\max}$ has no feasible solution.
This would induce, however, the same $O(n^2)$ complexity of the algorithm provided in \cite{hjm12} for problem $F2 | no-idle,no-wait | \cal{G}$.  
\bigskip

\noindent
To reach a better complexity, we strongly exploit the specific no-idle/no-wait constraint that strictly links
%The peculiarity of the no-idle, no-wait effect strictly links 
the $F2 | no-idle, no-wait | C_{\max}$ problem 
to the game of dominoes.
Dominoes are $1$ x $2$ rectangular tiles with each $1$ x $1$ square marked with spots
indicating a number. A traditional set of dominoes consists of all 28 unordered pairs of numbers between
0 and 6. We refer here to the generalization of dominoes presented in \cite{dmw14}
in which $n$ tiles are present, each of the tiles can have
any integer (or symbol) on each end and not necessarily all pairs of numbers are present.

\begin{figure*}[!ht]
    \begin{center}
        \includegraphics[height=2.5cm,clip=true]{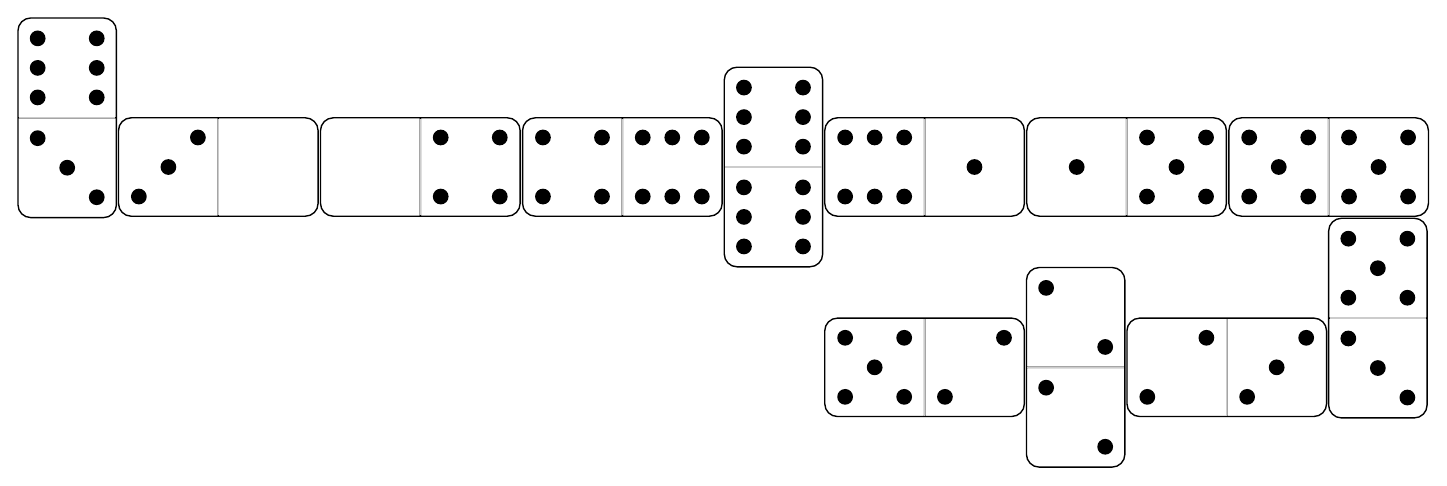}
%\vspace*{-1.2cm}				
        \caption[Fig]{Solution of a Single Player Dominoes problem with 12 dominoes}
        \label{domino}
    \end{center}
\end{figure*}

\bigskip
\noindent
In \cite{dmw14}, it is shown that the Single Player Dominoes ($SPD$) problem (where a single player 
tries to lay down all dominoes in a chain with the numbers matching at each adjacency)
is polynomially solvable as it can be seen as the solution of an Eulerian path problem
on an undirected multigraph (in a way similar to the result provided in \cite{hjm12} but requiring lower complexity). 
Figure \ref{domino} shows the solution of an $SPD$ problem with 12 tiles with numbers 
included between $0$ and $6$.

\noindent
We refer here to the oriented version of $SPD$, denoted by $OSPD$, where all dominoes
have an orientation, e.g. if the numbers are $i$ and $j$, only the orientation $i \rightarrow j$ is allowed
but not viceversa. The following Lemma holds.

\begin{Lemma}\label{ospd}
Problem $OSPD$ is polynomially solvable.
\end{Lemma}

\begin{proof}
Following the proof for problem $SPD$ in \cite{dmw14}, we construct a multigraph $G$ where vertices 
correspond to numbers and arcs (instead of edges) correspond to oriented dominoes. The rest works similarly.
Hence, in this case, the problem reduces to finding an Eulerian directed path in a directed multigraph
which in turn is polynomially solvable \cite{FH91}.
%A connected digraph G contains
%a directed eulerian trail if and only if G
%contains vertices u and v such that out-deg(u) = in-deg(u) + 1,
%out-deg(v) = in-deg(u)-1
%and out-deg(x) = in-deg(x) for all other vertices x of G.
%The trail begins at u and ends at v
\end{proof}

\noindent
The following proposition holds.

\begin{prop}\label{P2}
$F2 | no-idle, no-wait | C_{\max}$ polynomially reduces to $OSPD$.
%$F2 | no-idle, no-wait | C_{\max}$ $\propto$ $OSPD$.
%Correspondingly, problem $F2 | no-idle, no-wait | C_{\max}$ is polynomially solvable.
\end{prop}
\begin{proof}
By generating for each job $J_j$ a related domino $\{p_{1,j},p_{2,j}\}$, we know that any complete sequence of oriented dominoes in $OSPD$
corresponds to a feasible sequence for $F2 | no-idle, no-wait | C_{\max}$ and viceversa. 
But then, due to Lemma \ref{l0}, the jobs processing times either respect case 1 or case 2 of condition (C2).
In case 1, the related sequence is optimal for $F2 | no-idle, no-wait | C_{\max}$
as the processing time of the first job on the first machine is given and correspondingly, due to Lemma \ref{l1},
the makespan is given.
In case 2, we simply rotate the sequence in order to start with a job $J_k$ having the smallest processing time 
on the first machine, that is job $J_k$ with $p_{1,k} = \min_{i=1,..,n} p_{1,i}$. Correspondingly, the makespan $C_{\max}= p_{1,k}+\sum_{i=1}^n p_{2,i}$ is 
minimum and the related solution is optimal.
\end{proof}

\bigskip

%{\bf Introduce here some comments on reference \ref{hjm12}}

%%%%% Begin Zone VTK
%Starting from the reduction stated in Proposition \ref{P2}, we propose an optimal algorithm, referred to as \texttt{AlgF2} 
%(Figure \ref{refAlgF2}), which solves problem $F2 | no-idle, no-wait | C_{\max}$ in linear time.
By means of the reduction in Proposition \ref{P2}, we propose an optimal algorithm, referred to as \texttt{AlgF2}, which solves problem $F2 | no-idle, no-wait | C_{\max}$ in linear time.

\begin{algorithm}[htp]
\SetAlgoRefName{AlgF2}
\SetAlgoLined
\KwData{A set of $n$ jobs to be scheduled with processing times $p_{1,i}$ and $p_{2,i}$, the associated multigraph $G$}
\KwResult{An optimal no-wait no-idle schedule, if it exists}
${\cal V}=\{p_{1,i}, p_{2,i} / i=1,...,n\}$\;
$d^+(\alpha_i)=|\{j / p_{1,j}=\alpha_i\}|$, $\forall \alpha_i\in{\cal V}$ \;
$d^-(\alpha_i)=|\{j / p_{2,j}=\alpha_i\}|$, $\forall \alpha_i\in{\cal V}$ \;             
$d(\alpha_i)=d^+(\alpha_i)-d^-(\alpha_i)$, $\forall \alpha_i\in{\cal V}$ \;
${\cal S}=\{\alpha_i / d(\alpha_i)\neq 0\}$ \;
\eIf{($|\cal S|$=0)}
	{Compute an Eulerian walk $ew$ in $G$ starting from vertex $\alpha_i=\min_{\alpha_k \in {\cal V}}(\alpha_k)$\;}
	{\eIf{($|{\cal S}|=2$ and ($d(\alpha_1)=1$ and $d(\alpha_2)=-1, \alpha_1, \alpha_2\in {\cal S}$))}
	    {Compute an Eulerian walk $ew$ in $G$ from vertex $\alpha_1$ to vertex $\alpha_2$\;}
			{Exit: Problem infeasible\;}
	}
$s^*$ is the optimal schedule, with $s^*[k]$ the job corresponding to arc $ew[k]$, $\forall k=1,...,n$\;
\Return{$s^*$}
\caption{Solving the $F2 | no-idle, no-wait | C_{\max}$ problem}
\label{refAlgF2}
\end{algorithm}

\begin{prop}\label{refpropAlgF2}
Algorithm \texttt{AlgF2} solves problem $F2 | no-idle, no-wait | C_{\max}$ in $O(n)$ time. 
\end{prop}
\begin{proof}
Algorithm \ref{refAlgF2} takes upon entry the set of $n$ jobs and its associated multigraph $G$ built as described in Lemma \ref{ospd}: each job processing time $p_{1,i}$ or $p_{2,i}$ corresponds to a vertex and there exists an arc from vertex $k$ to vertex $\ell$ iff there exists a job $J_i$ such that $p_{1,i}=k$ and $p_{2,i}=\ell$. Then, computing a feasible schedule to the $F2 | no-idle, no-wait | C_{\max}$ problem reduces to computing, if it exists,  a directed Eulerian path in $G$. The computation of an optimal solution is done by Algorithm \ref{refAlgF2} directly by exploiting existence conditions of such a path.\\
Concerning the running time of the algorithm, it can be noticed that lines 1-5 can be executed in $O(n)$ time by an appropriate implementation. Lines 7 and 10 require to compute an Eulerian path in an oriented graph with $n$ arcs, which can be done in $O(n)$ time 
(\cite{FH91}). 
%By the way, the overall algorithm requires $O(n)$ time and polynomial space.\\
Notice that the construction of graph $G$, needed as an input to the algorithm, can also be built in $O(n)$ time.
\end{proof}
%%%%% End Zone VTK

\noindent
Consider the $9$-job example of Table \ref{tab1}.

\begin{table*}[htbp]
	\centering
	\footnotesize
		\begin{tabular}{||c||c|c|c|c|c|c|c|c|c||}
			\hline
	   	 $i$  &  $J_1$   &  $J_2$   &  $J_3$   &  $J_4$   &  $J_5$   &  $J_6$   &  $J_7$   &  $J_8$   &  $J_9$ \\ \hline
 $p_{1,i}$  &   $5$    &  $3$     &  $4$     &  $6$     &  $1$     &  $5$     &  $3$     &  $2$     &  $4$   \\ \hline
 $p_{2,i}$  &   $3$    &  $4$     &  $6$     &  $1$     &  $5$     &  $3$     &  $2$     &  $4$     &  $5$   \\ \hline
 	\end{tabular}
%		\begin{tabular}{||c||c|c||}
%			\hline
%		$i$  &   $p_{1,i}$  &  $p_{2,i}$ \\ \hline
% $J_1$   &   $5$        &  $3$ \\ \hline
% $J_2$   &   $3$        &  $4$ \\ \hline
% $J_3$   &   $4$        &  $6$ \\ \hline
% $J_4$   &   $6$        &  $1$ \\ \hline
% $J_5$   &   $1$        &  $5$ \\ \hline
% $J_6$   &   $5$        &  $3$ \\ \hline
% $J_7$   &   $3$        &  $2$ \\ \hline
% $J_8$   &   $2$        &  $4$ \\ \hline
% $J_9$   &   $4$        &  $5$ \\ \hline
%		\end{tabular}
	\caption{A 9-job instance of problem $F2 | no-idle, no-wait | C_{\max}$}
	\label{tab1}
\end{table*}

\noindent
The corresponding optimal solution is provided in 
Figure \ref{f2_picture}.

%\vspace*{-.8cm}
%\begin{figure*}[!ht]
%    \begin{center}
%        \includegraphics[height=12.5cm,clip=true]{picture_f2sol.pdf}
%\vspace*{-6.8cm}				
%        \caption[Fig]{The optimal solution of the problem in Table \ref{tab1}}
%        \label{f2_picture}
%    \end{center}
%\end{figure*}

\begin{figure*}[!ht]
    \begin{center}
        \ifx\JPicScale\undefined\def\JPicScale{1}\fi
\unitlength \JPicScale mm
\begin{picture}(113,25)(0,0)
\linethickness{0.3mm}
\put(4,6){\line(1,0){104}}
\linethickness{0.3mm}
\put(4,5){\line(0,1){1}}
\linethickness{0.3mm}
\put(19,5){\line(0,1){1}}
\put(-0.25,17){\makebox(0,0)[cc]{$M_1$}}

\put(-0.25,8){\makebox(0,0)[cc]{$M_2$}}

\put(34,3.38){\makebox(0,0)[cc]{10}}

\put(106,3.38){\makebox(0,0)[cc]{34}}

\linethickness{0.3mm}
\put(106,5){\line(0,1){3}}
\put(43,17){\makebox(0,0)[cc]{$P_{1,9}$}}

\put(34,17){\makebox(0,0)[cc]{$P_{1,8}$}}

\put(26,17){\makebox(0,0)[cc]{$P_{1,7}$}}

\put(15,17){\makebox(0,0)[cc]{$P_{1,6}$}}

\put(5,25){\makebox(0,0)[cc]{$P_{1,5}$}}

\put(94,17){\makebox(0,0)[cc]{$P_{1,4}$}}

\put(79,17){\makebox(0,0)[cc]{$P_{1,3}$}}

\put(69,17){\makebox(0,0)[cc]{$P_{1,2}$}}

\put(57,17){\makebox(0,0)[cc]{$P_{1,1}$}}

\linethickness{0.3mm}
\put(5,20){\line(0,1){3}}
\put(5,20){\vector(0,-1){0.12}}
\put(57,8){\makebox(0,0)[cc]{$P_{2,9}$}}

\put(27,8){\makebox(0,0)[cc]{$P_{2,6}$}}

\put(43,8){\makebox(0,0)[cc]{$P_{2,8}$}}

\put(34,8){\makebox(0,0)[cc]{$P_{2,7}$}}

\put(15,8){\makebox(0,0)[cc]{$P_{2,5}$}}

\put(113,16){\makebox(0,0)[cc]{$P_{2,4}$}}

\put(68,8){\makebox(0,0)[cc]{$P_{2,1}$}}

\put(79,8){\makebox(0,0)[cc]{$P_{2,2}$}}

\put(94,8){\makebox(0,0)[cc]{$P_{2,3}$}}

\linethickness{0.3mm}
\multiput(105,10)(0.15,0.12){33}{\line(1,0){0.15}}
\put(105,10){\vector(-4,-3){0.12}}
\linethickness{0.3mm}
\put(7,11){\line(1,0){15}}
\put(7,6){\line(0,1){5}}
\put(22,6){\line(0,1){5}}
\put(7,6){\line(1,0){15}}
\linethickness{0.3mm}
\put(22,11){\line(1,0){9}}
\put(22,6){\line(0,1){5}}
\put(31,6){\line(0,1){5}}
\put(22,6){\line(1,0){9}}
\linethickness{0.3mm}
\put(31,11){\line(1,0){6}}
\put(31,6){\line(0,1){5}}
\put(37,6){\line(0,1){5}}
\put(31,6){\line(1,0){6}}
\linethickness{0.3mm}
\put(37,11){\line(1,0){12}}
\put(37,6){\line(0,1){5}}
\put(49,6){\line(0,1){5}}
\put(37,6){\line(1,0){12}}
\linethickness{0.3mm}
\put(49,11){\line(1,0){15}}
\put(49,6){\line(0,1){5}}
\put(64,6){\line(0,1){5}}
\put(49,6){\line(1,0){15}}
\linethickness{0.3mm}
\put(64,11){\line(1,0){9}}
\put(64,6){\line(0,1){5}}
\put(73,6){\line(0,1){5}}
\put(64,6){\line(1,0){9}}
\linethickness{0.3mm}
\put(73,11){\line(1,0){12}}
\put(73,6){\line(0,1){5}}
\put(85,6){\line(0,1){5}}
\put(73,6){\line(1,0){12}}
\linethickness{0.3mm}
\put(85,11){\line(1,0){18}}
\put(85,6){\line(0,1){5}}
\put(103,6){\line(0,1){5}}
\put(85,6){\line(1,0){18}}
\put(19,3.38){\makebox(0,0)[cc]{5}}

\linethickness{0.3mm}
\put(34,5){\line(0,1){1}}
\linethickness{0.3mm}
\put(49,5){\line(0,1){1}}
\put(49,3.38){\makebox(0,0)[cc]{15}}

\linethickness{0.3mm}
\put(64,5){\line(0,1){1}}
\put(64,3.38){\makebox(0,0)[cc]{20}}

\put(79,3.38){\makebox(0,0)[cc]{25}}

\linethickness{0.3mm}
\put(79,5){\line(0,1){1}}
\linethickness{0.3mm}
\put(94,5){\line(0,1){1}}
\put(94,3.38){\makebox(0,0)[cc]{30}}

\linethickness{0.3mm}
\put(103,11){\line(1,0){3}}
\put(103,6){\line(0,1){5}}
\put(106,6){\line(0,1){5}}
\put(103,6){\line(1,0){3}}
\linethickness{0.3mm}
\put(4,20){\line(1,0){3}}
\put(4,15){\line(0,1){5}}
\put(7,15){\line(0,1){5}}
\put(4,15){\line(1,0){3}}
\linethickness{0.3mm}
\put(7,20){\line(1,0){15}}
\put(7,15){\line(0,1){5}}
\put(22,15){\line(0,1){5}}
\put(7,15){\line(1,0){15}}
\linethickness{0.3mm}
\put(22,20){\line(1,0){9}}
\put(22,15){\line(0,1){5}}
\put(31,15){\line(0,1){5}}
\put(22,15){\line(1,0){9}}
\linethickness{0.3mm}
\put(31,20){\line(1,0){6}}
\put(31,15){\line(0,1){5}}
\put(37,15){\line(0,1){5}}
\put(31,15){\line(1,0){6}}
\linethickness{0.3mm}
\put(37,20){\line(1,0){12}}
\put(37,15){\line(0,1){5}}
\put(49,15){\line(0,1){5}}
\put(37,15){\line(1,0){12}}
\linethickness{0.3mm}
\put(49,20){\line(1,0){15}}
\put(49,15){\line(0,1){5}}
\put(64,15){\line(0,1){5}}
\put(49,15){\line(1,0){15}}
\linethickness{0.3mm}
\put(64,20){\line(1,0){9}}
\put(64,15){\line(0,1){5}}
\put(73,15){\line(0,1){5}}
\put(64,15){\line(1,0){9}}
\linethickness{0.3mm}
\put(73,20){\line(1,0){12}}
\put(73,15){\line(0,1){5}}
\put(85,15){\line(0,1){5}}
\put(73,15){\line(1,0){12}}
\linethickness{0.3mm}
\put(85,20){\line(1,0){18}}
\put(85,15){\line(0,1){5}}
\put(103,15){\line(0,1){5}}
\put(85,15){\line(1,0){18}}
\put(7,0){\makebox(0,0)[cc]{$C_{\max}=34$}}

\put(7,0){\makebox(0,0)[cc]{}}

\end{picture}			
        \caption[Fig]{The optimal solution of the problem of Table \ref{tab1}}
        \label{f2_picture}
    \end{center}
\end{figure*}
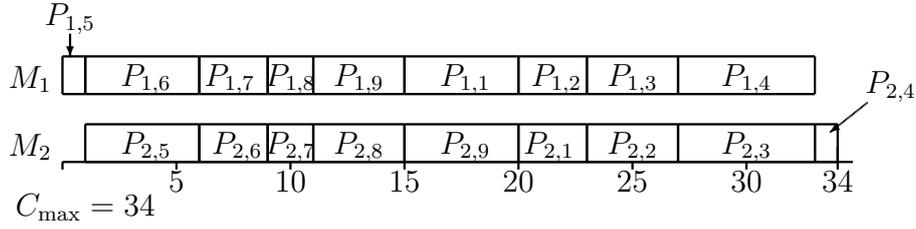

\bigskip

\noindent
The input data and the solution of the $OSPD$ problem corresponding to the $F2 | no-idle, no-wait | C_{\max}$ instance of Table \ref{tab1}
are provided in Figure \ref{ospdfig}.

\begin{figure*}[!ht]
    \begin{center}
        \includegraphics[height=7.5cm,clip=true]{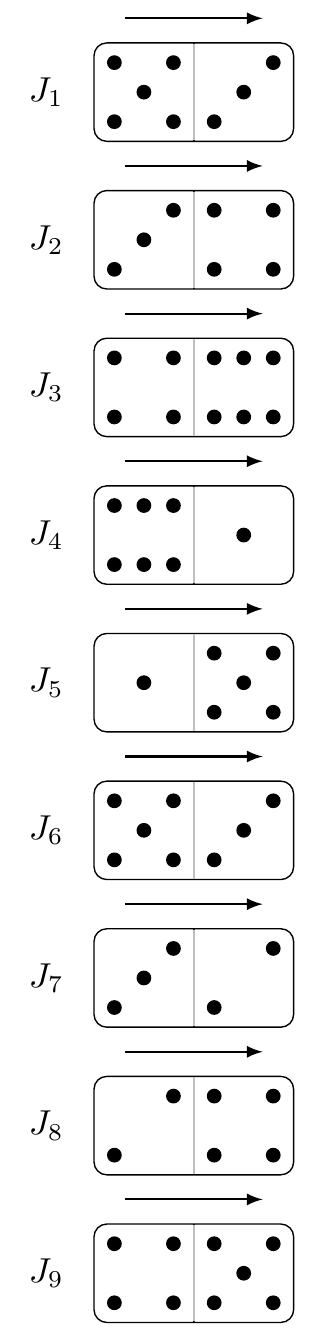} \hspace*{3cm}
        \includegraphics[height=4.5cm,clip=true]{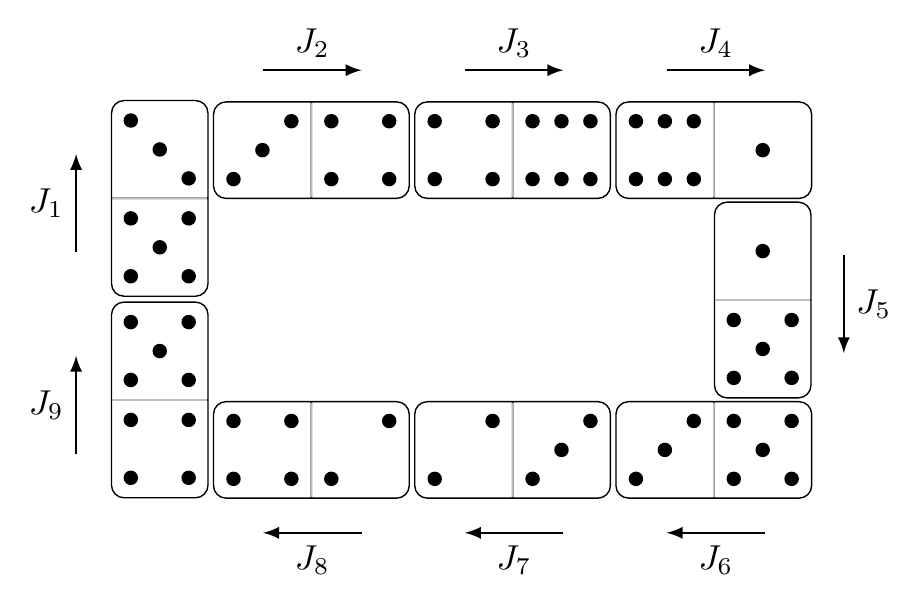}
\vspace*{-0.2cm}				
        \caption[Fig]{The dominoes corresponding to the flow shop instance of Table \ref{tab1} and
				the related $OSPD$ problem solution}
        \label{ospdfig}
    \end{center}
\end{figure*}

%\bigskip

%\noindent
%The solution of the $OSPD$ problem of Figure \ref{ospd} is depicted in Figure \ref{ospd1}.

%\begin{figure*}[!ht]
%    \begin{center}
%        \includegraphics[height=3.5cm,clip=true]{dominoNew.pdf}
%\vspace*{-0.2cm}				
%        \caption[Fig]{The solution of the $OSPD$ problem of Figure \ref{ospd}}
%        \label{ospd1}
%    \end{center}
%\end{figure*}

\bigskip

\noindent
The multigraph computed by Algorithm \ref{refAlgF2}  on the flow shop instance of Table \ref{tab1}
is depicted in Figure \ref{euler}.

\begin{figure*}[!ht]
    \begin{center}
        \includegraphics[height=3.5cm,clip=true]{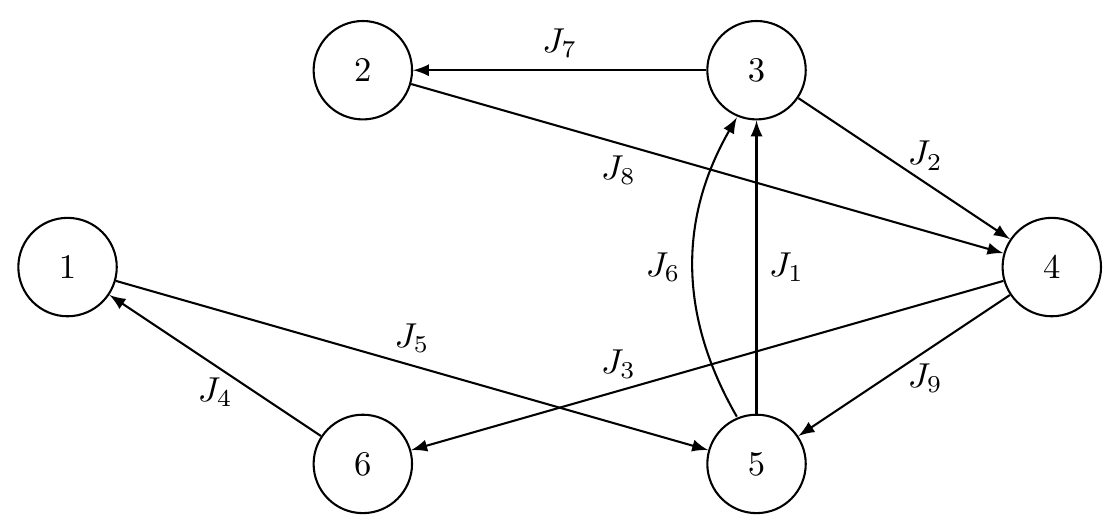}
\vspace*{-0.2cm}				
        \caption[Fig]{The multigraph originated by Algorithm \ref{refAlgF2} 
				on the flow shop instance of Table \ref{tab1}}
        \label{euler}
    \end{center}
\end{figure*}

%\bigskip
%\noindent

\subsection{The $F2 | no-idle, no-wait | C_{\max}$ problem and the Common/Distinct Successors Directed Hamiltonian Path problem.}

\noindent
Problem $F2 | no-idle, no-wait | C_{\max}$ is also linked to a special case of the Hamiltonian Path problem
on a connected digraph.
Consider a connected digraph $G(V,A)$ 
that has the following property:  
$\forall v_i,v_j \in V$, either 
$S_i\cap S_j = \emptyset$, or $S_i=S_j$,
where $S_i$ denotes the set of successors of vertex $v_i$.
In other words, each pair of vertices either has no common successors or has all successors in common. Let indicate the Hamiltonian path problem in that graph as the Common/Distinct Successors Directed Hamiltonian Path (CDSDHP) problem. 
Notice that problem CDSDHP may have no feasible solution as indicated in Figure 
\ref{dominotspnosol}.
 
\begin{figure*}[!ht]
    \begin{center}
        \includegraphics[height=4cm,clip=true]{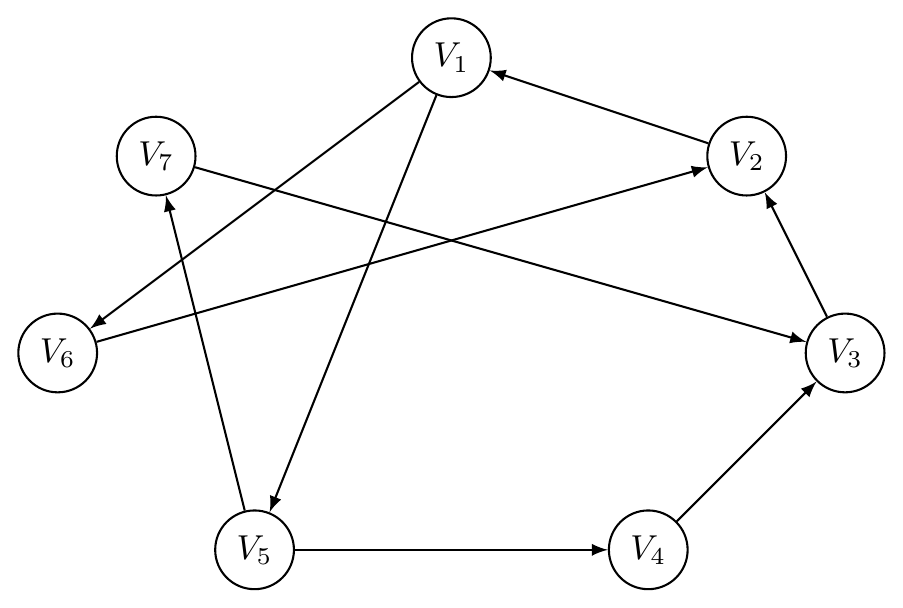}
\vspace*{-0.2cm}				
        \caption[Fig]{A non feasible instance of the CDSDHP problem}
        \label{dominotspnosol}
    \end{center}
\end{figure*}

\noindent
The following straightforward lemma holds.
\begin{Lemma}\label{cdshdp}
Problem OSPD polynomially reduces to problem CDSDHP.
\end{Lemma}
\begin{proof}
Given an OSPD problem with $n$ tiles, it is sufficient to generate a digraph $G(V,A)$ 
where each oriented tile corresponds to a vertex and there is an arc between two vertices 
if the corresponding tiles can match.
Then, any complete sequence of oriented dominoes in OSPD
corresponds to an hamiltonian path in the digraph and viceversa.
\end{proof}

\begin{figure*}[!ht]
    \begin{center}
        \includegraphics[height=3.5cm,clip=true]{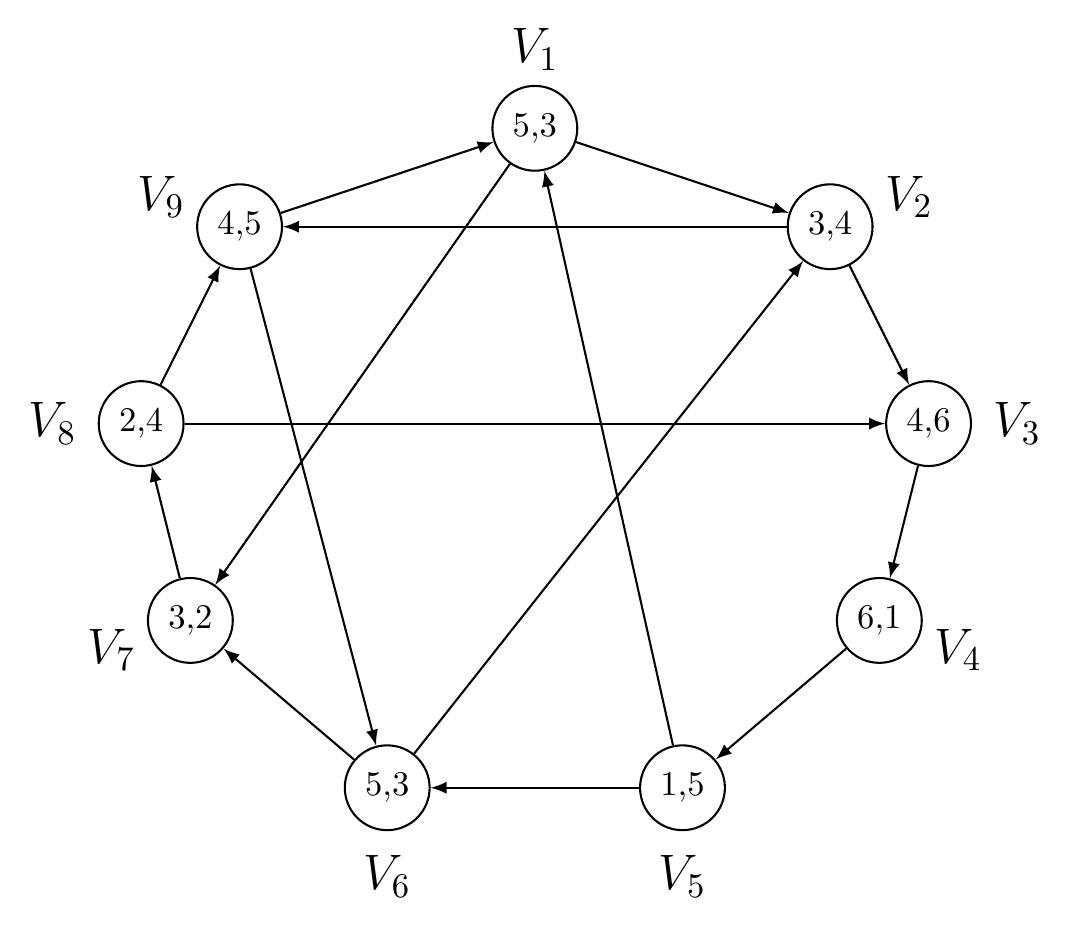}
\vspace*{-0.2cm}				
        \caption[Fig]{The CDSDHP instance corresponding to the OSPD problem of Figure \ref{ospdfig} and to the flow shop instance of Table \ref{tab1}}
        \label{dominotsp2}
    \end{center}
\end{figure*}

\noindent
From Property \ref{P2} and Lemma \ref{cdshdp}, we know that 
$F2 | no-idle, no-wait | C_{\max}$ $\propto$ OSPD $\propto$ CDSDHP.
Figure \ref{dominotsp2}
depicts the CDSDHP instance corresponding to the OSPD problem of Figure \ref{ospdfig}
and to the flow shop instance of Table \ref{tab1}. The polynomial reduction between these problems actually work also in the opposite sense 
as indicated by the following Proposition \ref{P4}. 

\noindent
For any given instance of problem CDSDHP, consider algorithm \texttt{AlgGenerF2}
which constructs a related instance of problem $F2 | no-idle, no-wait | C_{\max}$.

\begin{algorithm}[htp]
\SetAlgoRefName{AlgGenerF2}
\SetAlgoLined
\KwData{A digraph $G(V,A)$ such that $\forall v_i,v_j \in V$, either 
$S_i\cap S_j = \emptyset$, or $S_i=S_j$}
\KwResult{An instance of problem $F2 | no-idle, no-wait | C_{\max}$ with $n$ jobs where $n=|V|$}
$\forall v_j \in V$ generate a corresponding job $J_j$\;
$count=1$\;
\While{$\exists v_k: p_{1,k}$ or $p_{2,k}$ have not yet been determined} 
{\eIf{($v_k$ has a self-loop)}
	{$p_{1,k}=p_{2,k}=count$; $count=count+1$\;}
	{\eIf{(both $p_{1,k}$ and $p_{2,k}$ have not yet been determined)\;}
		{$p_{1,k}=count$; $p_{2,k}=count+1$; $count=count+2$\;} 
		{\eIf{($p_{1,k}$ has not yet been determined)\;}
			{$p_{1,k}=count$; $count=count+1$\;} 
				{\If{($p_{2,k}$ has not yet been determined)\;}
					{$p_{2,k}=count$; $count=count+1$\;} 				
					
    }
  }
}
$\forall v_j $ successor of $v_k$, $\; p_{1,j}=p_{2,k}$\; 
$\forall v_j $ predecessor of $v_k$, $\; p_{2,j}=p_{1,k}$\; 
$\forall v_j$ having common successors with $v_k$, $p_{2,j} =p_{2,k}$\;
$\forall v_j$ having common predecessors with $v_k$, $p_{1,j} =p_{1,k}$\;
}
\Return{}
\caption{Generating an instance of problem $F2 | no-idle, no-wait | C_{\max}$}
\label{refAlgGenerF2}
\end{algorithm}

\noindent
The following proposition holds.

\begin{prop}\label{P4}
CDSDHP polynomially reduces to problem $F2 | no-idle,no-wait | C_{\max}$.
%$CDSHDP$ $\propto$  $F2 | no-idle,no-wait | C_{\max}$.
%Correspondingly, problem $CDSHDP$  is polynomially solvable.
\end{prop} 
\begin{proof}
For any given instance of CDSDHP, we generate an instance of
$F2 | no-idle,no-wait | C_{\max}$ according to Algorithm \texttt{AlgGenerF2}.
Notice that  the jobs processing times 
are generated in such a way that, if there is an arc from $v_i$ to $v_j$,
then, we have $p_{2,i}=p_{1,j}$. 
If a feasible sequence of $F2 | no-idle,no-wait | C_{\max}$ exists,
then, for each pair of consecutive jobs $i,j$ with $i$ preceding $j$,
we have $p_{2,i}=p_{1,j}$ and, correspondingly, there is an arc from $v_i$ to $v_j$.
Thus, the corresponding sequence of vertices constitutes an Hamiltonian directed path
for the considered instance of CDSDHP.
Conversely, if an Hamiltonian directed path exists
for the considered instance of CDSDHP,
the corresponding sequence of jobs in problem $F2 | no-idle,no-wait | C_{\max}$ is also feasible.
\end{proof}

\begin{Remark}
Problem CDSDHP is solvable in $O(n^2)$ time.
Indeed, the {\em while}-loop in Algorithm \ref{refAlgGenerF2} is applied at most $O(n)$ times and the predecessors and successors
of any given $v_k$ are at most $O(n)$. Correspondingly, Algorithm \ref{refAlgGenerF2}
has $O(n^2)$ complexity, while, from Proposition {refpropAlgF2}, we know that problem 
$F2 | no-idle, no-wait | C_{\max}$ is solvable in linear time.
\end{Remark}

\section{$m$-machine no-idle no-wait flow shop scheduling}
\label{fmss}

%%%%% Begin Zone VTK
In this section we focus on the $m$-machine no-wait no-idle flowshop problem. Here 
each job $j$ has processing times $p_{i,j}$, for all machines $M_1, ..., M_m$.
Let us first introduce a generalized version of Lemma \ref{l0}
illustrated in Figure \ref{reffigFM1}.

%\begin{figure*}[!ht]
%    \begin{center}
%        \includegraphics[height=3.5cm,clip=true]{GanttFM.png}
%        \caption[Fig]{Solution of an $m$-machine flowshop problem}
%        \label{reffigFM1}
%    \end{center}
%\end{figure*}

\begin{figure*}[!ht]
    \begin{center}
        \ifx\JPicScale\undefined\def\JPicScale{1}\fi
\unitlength \JPicScale mm
\begin{picture}(127,48)(0,0)
\linethickness{0.3mm}
\put(8,34){\line(1,0){104}}
\put(2,46){\makebox(0,0)[cc]{$M_1$}}

\put(2,37){\makebox(0,0)[cc]{$M_2$}}

\linethickness{0.3mm}
\put(11,39){\line(1,0){15}}
\put(11,34){\line(0,1){5}}
\put(26,34){\line(0,1){5}}
\put(11,34){\line(1,0){15}}
\linethickness{0.3mm}
\put(26,39){\line(1,0){9}}
\put(26,34){\line(0,1){5}}
\put(35,34){\line(0,1){5}}
\put(26,34){\line(1,0){9}}
\linethickness{0.3mm}
\put(35,39){\line(1,0){6}}
\put(35,34){\line(0,1){5}}
\put(41,34){\line(0,1){5}}
\put(35,34){\line(1,0){6}}
\linethickness{0.3mm}
\put(53,39){\line(1,0){15}}
\put(53,34){\line(0,1){5}}
\put(68,34){\line(0,1){5}}
\put(53,34){\line(1,0){15}}
\linethickness{0.3mm}
\put(68,39){\line(1,0){9}}
\put(68,34){\line(0,1){5}}
\put(77,34){\line(0,1){5}}
\put(68,34){\line(1,0){9}}
\linethickness{0.3mm}
\put(8,48){\line(1,0){3}}
\put(8,43){\line(0,1){5}}
\put(11,43){\line(0,1){5}}
\put(8,43){\line(1,0){3}}
\linethickness{0.3mm}
\put(11,48){\line(1,0){15}}
\put(11,43){\line(0,1){5}}
\put(26,43){\line(0,1){5}}
\put(11,43){\line(1,0){15}}
\linethickness{0.3mm}
\put(26,48){\line(1,0){9}}
\put(26,43){\line(0,1){5}}
\put(35,43){\line(0,1){5}}
\put(26,43){\line(1,0){9}}
\linethickness{0.3mm}
\put(41,48){\line(1,0){12}}
\put(41,43){\line(0,1){5}}
\put(53,43){\line(0,1){5}}
\put(41,43){\line(1,0){12}}
\linethickness{0.3mm}
\put(53,48){\line(1,0){15}}
\put(53,43){\line(0,1){5}}
\put(68,43){\line(0,1){5}}
\put(53,43){\line(1,0){15}}
\put(11,28){\makebox(0,0)[cc]{}}

\linethickness{0.3mm}
\put(8,25){\line(1,0){104}}
\put(2,28){\makebox(0,0)[cc]{$M_3$}}

\linethickness{0.3mm}
\put(26,30){\line(1,0){9}}
\put(26,25){\line(0,1){5}}
\put(35,25){\line(0,1){5}}
\put(26,25){\line(1,0){9}}
\linethickness{0.3mm}
\put(35,30){\line(1,0){6}}
\put(35,25){\line(0,1){5}}
\put(41,25){\line(0,1){5}}
\put(35,25){\line(1,0){6}}
\linethickness{0.3mm}
\put(41,30){\line(1,0){12}}
\put(41,25){\line(0,1){5}}
\put(53,25){\line(0,1){5}}
\put(41,25){\line(1,0){12}}
\linethickness{0.3mm}
\put(68,30){\line(1,0){9}}
\put(68,25){\line(0,1){5}}
\put(77,25){\line(0,1){5}}
\put(68,25){\line(1,0){9}}
\linethickness{0.3mm}
\put(77,30){\line(1,0){12}}
\put(77,25){\line(0,1){5}}
\put(89,25){\line(0,1){5}}
\put(77,25){\line(1,0){12}}
\linethickness{0.3mm}
\put(8,9){\line(1,0){118}}
\put(2,12){\makebox(0,0)[cc]{$M_{m-1}$}}

\linethickness{0.3mm}
\put(41,14){\line(1,0){12}}
\put(41,9){\line(0,1){5}}
\put(53,9){\line(0,1){5}}
\put(41,9){\line(1,0){12}}
\linethickness{0.3mm}
\put(53,14){\line(1,0){15}}
\put(53,9){\line(0,1){5}}
\put(68,9){\line(0,1){5}}
\put(53,9){\line(1,0){15}}
\linethickness{0.3mm}
\put(68,14){\line(1,0){9}}
\put(68,9){\line(0,1){5}}
\put(77,9){\line(0,1){5}}
\put(68,9){\line(1,0){9}}
\linethickness{0.3mm}
\put(89,14){\line(1,0){18}}
\put(89,9){\line(0,1){5}}
\put(107,9){\line(0,1){5}}
\put(89,9){\line(1,0){18}}
\put(11,3){\makebox(0,0)[cc]{}}

\linethickness{0.3mm}
\put(8,0){\line(1,0){119}}
\put(2,3){\makebox(0,0)[cc]{$M_m$}}

\linethickness{0.3mm}
\put(53,5){\line(1,0){15}}
\put(53,0){\line(0,1){5}}
\put(68,0){\line(0,1){5}}
\put(53,0){\line(1,0){15}}
\linethickness{0.3mm}
\put(68,5){\line(1,0){9}}
\put(68,0){\line(0,1){5}}
\put(77,0){\line(0,1){5}}
\put(68,0){\line(1,0){9}}
\linethickness{0.3mm}
\put(77,5){\line(1,0){12}}
\put(77,0){\line(0,1){5}}
\put(89,0){\line(0,1){5}}
\put(77,0){\line(1,0){12}}
\linethickness{0.3mm}
\put(107,5){\line(1,0){10}}
\put(107,0){\line(0,1){5}}
\put(117,0){\line(0,1){5}}
\put(107,0){\line(1,0){10}}
\put(3,20){\makebox(0,0)[cc]{$\vdots$}}

\put(9,46){\makebox(0,0)[cc]{1}}

\put(18,46){\makebox(0,0)[cc]{2}}

\put(31,46){\makebox(0,0)[cc]{3}}

\put(47,46){\makebox(0,0)[cc]{$n-1$}}

\put(61,46){\makebox(0,0)[cc]{$n$}}

\put(18,37){\makebox(0,0)[cc]{1}}

\put(31,37){\makebox(0,0)[cc]{2}}

\put(38,37){\makebox(0,0)[cc]{3}}

\put(60,37){\makebox(0,0)[cc]{$n-1$}}

\put(73,37){\makebox(0,0)[cc]{$n$}}

\put(31,28){\makebox(0,0)[cc]{1}}

\put(38,28){\makebox(0,0)[cc]{2}}

\put(47,28){\makebox(0,0)[cc]{3}}

\put(72,28){\makebox(0,0)[cc]{$n-1$}}

\put(83,28){\makebox(0,0)[cc]{$n$}}

\put(47,12){\makebox(0,0)[cc]{1}}

\put(61,3){\makebox(0,0)[cc]{1}}

\put(61,12){\makebox(0,0)[cc]{2}}

\put(72,3){\makebox(0,0)[cc]{2}}

\put(72,12){\makebox(0,0)[cc]{3}}

\put(83,3){\makebox(0,0)[cc]{3}}

\linethickness{0.3mm}
\put(107,14){\line(1,0){10}}
\put(107,9){\line(0,1){5}}
\put(117,9){\line(0,1){5}}
\put(107,9){\line(1,0){10}}
\linethickness{0.3mm}
\put(117,5){\line(1,0){9}}
\put(117,0){\line(0,1){5}}
\put(126,0){\line(0,1){5}}
\put(117,0){\line(1,0){9}}
\put(99,12){\makebox(0,0)[cc]{$n-1$}}

\put(112,12){\makebox(0,0)[cc]{$n$}}

\put(112,3){\makebox(0,0)[cc]{$n-1$}}

\put(122,3){\makebox(0,0)[cc]{$n$}}

\put(98,3){\makebox(0,0)[cc]{$\ldots$}}

\put(83,12){\makebox(0,0)[cc]{$\ldots$}}

\put(60,28){\makebox(0,0)[cc]{$\ldots$}}

\put(47,37){\makebox(0,0)[cc]{$\ldots$}}

\put(38,46){\makebox(0,0)[cc]{$\ldots$}}

\end{picture}
        \caption[Fig]{Solution of an $m$-machine flowshop problem}
        \label{reffigFM1}
    \end{center}
\end{figure*}
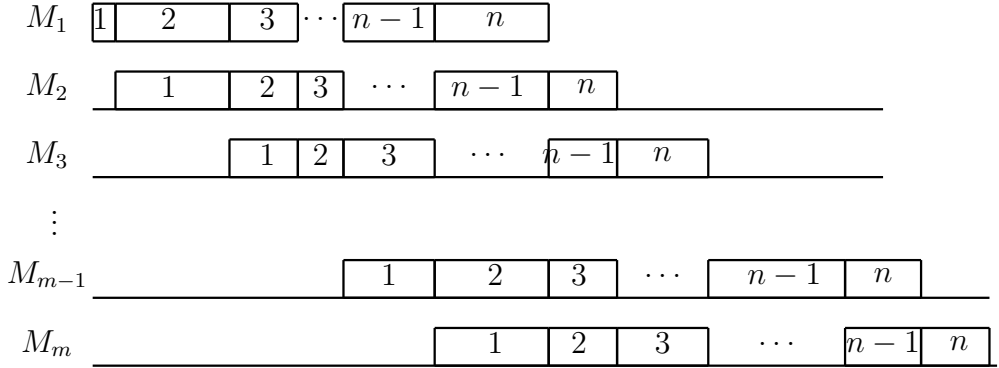

\begin{Lemma}\label{lFm1}
\noindent
\begin{enumerate}
\item[(C3)] A necessary condition to have a feasible solution for problem $F | no-idle, no-wait | C_{\max}$ is 
that there always exists an indexing of the jobs so that, $\forall j=1,...,m-1$, $p_{j+1,1},...,p_{j+1,n-1}$ and $p_{j,2},...,p_{j,n}$ constitute different permutations of the same vector of elements.
\item[(C4)] When the above condition (C3) holds, then
\begin{enumerate}
\item[Case 1] if $(p_{1,1}\neq p_{2,n}$ or $p_{2,1}\neq p_{3,n}$ or ... or $p_{m-1,1}\neq p_{m,n})$, every feasible sequence must have a job with processing times 
$(p_{1,1}, ..., p_{m-1,1})$ on machines $M_1$ to $M_{m-1}$ in first position and a job  with processing time $(p_{2,n}, ..., p_{m,n})$  on machines $M_2$ to $M_m$ in last position.
%Correspondingly the makespan is $a_1+p(B)$.
%\item[Case 2] if $(p_{1,1}=p_{2,n}$ and $p_{2,1}=p_{3,n}$ and ... and $p_{m-1,1}=p_{m,n})$ and there exists a feasible sequence,
\item[Case 2] if $(p_{1,1}=p_{2,n}$ and $p_{2,1}=p_{3,n}$ and ... and $p_{m-1,1}=p_{m,n})$
then there exists at least $n$ feasible sequences each starting with a different job by simply rotating the starting sequence
as in a cycle.
\end{enumerate}
\end{enumerate}
\end{Lemma}
\begin{proof}
Similar to that of Lemma \ref{l0}.
\end{proof}

From Lemma \ref{lFm1} and Figure \ref{reffigFM1} we can evince that in an optimal sequence if job $J_i$ immediately precedes job $J_k$ we must have $p_{j+1,i}=p_{j,k}, 
\forall j=1, ..., m-1$. Consequently, for a feasible subsequence $(J_\ell,J_i,J_k)$ we must have:
\begin{eqnarray}\label{myeq2}
[p_{2,\ell};...;p_{m,\ell}] = [p_{1,i};...;p_{m-1,i}] \text{ and } [p_{2,i};...;p_{m,i}] = [p_{1,k};...;p_{m-1,k}].
\end{eqnarray}
This can be represented in terms of dominoes as indicated in Figure \ref{reffigFM2}.

\begin{figure*}[!ht]
    \begin{center}
        \includegraphics[width=14cm]{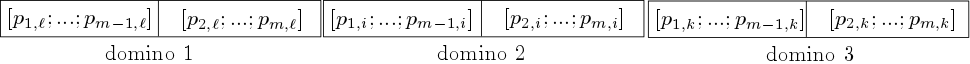}
        \caption[Fig]{Three consecutive jobs / dominos}
        \label{reffigFM2}
    \end{center}
\end{figure*}

\noindent The following proposition holds.

\begin{prop}\label{P5}
Problem $F | no-idle,no-wait | C_{\max}$ can be solved to optimality in 
%$O((m-1)n\log(n))$ time. 
$O(mn\log n)$ time.
\end{prop}	
\begin{proof}
We show that any instance of the $F | no-idle,no-wait | C_{\max}$ problem can be reduced  to another instance of problem
$F2 | no-idle,no-wait | C_{\max}$. The reduction works as follows.
From any instance $I_m$ of the m-machine problem, build $2n$ vectors $[p_{1,i};...;p_{m-1,i}]$ and $[p_{2,i};...;p_{m,i}], $ for each job $J_i\in I_m$. Sort these vectors according 
to the lexicographical increasing order of their values and let $v_{[t]}$ be the $t$-th vector in this order. This process requires 
$O((m-1)n\log n) \approx O(mn\log n)$ time to be done.\\
Then, we create an instance $I_2$ of the 2-machine problem: $\forall J_i\in I_m$, we create a new job $J_\ell \in I_2$ such that $p_{1,\ell}=t_1$ and $p_{2,\ell}=t_2$ with 
$v_{[t_1]}=[p_{1,i};...;p_{m-1,i}]$ and $v_{[t_2]}=[p_{2,i};...;p_{m,i}]$. This process can be done in $O((m-1)n\log n) \approx O(mn\log n)$ time.
\smallskip

\noindent
Finally, apply Algorithm \ref{refAlgF2} to $I_2$. If instance $I_2$ is infeasible, then there does not exist a feasible solution to $I_m$ due to equation \ref{myeq2} which cannot be 
answered for any triplet of consecutive jobs. If instance $I_2$ admits a feasible solution, then a feasible solution to $I_m$ can be easily derived since jobs in $I_2$ correspond to jobs in
$I_m$. Let $s$ be the sequence of jobs of $I_m$ obtained from the optimal solution returned by Algorithm \ref{refAlgF2}  on $I_2$. If Case 1 of (C4) holds for $s$ then $s$ is also
optimal for $I_m$. Otherwise, it means that Case 2 holds and then the optimal solution $s^*$ to $I_m$ is obtained by considering the circular permutation of jobs in $s$ where job $J_\ell$ with 
$\ell=argmin_{k=1...n}(\sum_{j=1}^{m-1} p_{j,k})$ is put first. The building of $s^*$ can be done in $O((m-1)n) \approx O(mn)$ time.
Consequently, the solution of the $F | no-idle,no-wait | C_{\max}$ problem can be done with an overall $O(mn\log n)$ time complexity. 
%$O((m-1)n\log(n))$.

\end{proof}

%%%%% End Zone VTK

\section{Two-machine no-idle no-wait job shop scheduling}
\label{jss}
In this section we consider problem $J2 | no-idle, no-wait | C_{\max}$.
In the job shop configuration, differently from the flow shop,  each job
has its own processing routing. Also, when jobs with different processing routes are present, for every feasible solution, the jobs sequences on the two machines  are necessarily different  due to the no-wait constraint on the jobs.
We establish the complexity of this problem 
%$J2 | no-idle, no-wait | C_{\max}$ 
%and 
%$O2 | no-idle, no-wait | C_{\max}$ 
%by proving their 
%$NP$-Hardness in the strong sense. 
by proving that it is
$NP$-Hard in the strong sense. 
To this extent 
we consider the Numerical Matching with Target Sums (NMTS) problem
that has been shown to be $NP$-complete in the strong sense \cite{GJ82}.

\bigskip

\noindent
{\bf NMTS} \\
{\bf Instance}: Disjoint sets $X$ and $Y$, each containing $m$ elements, a size $s(k) \in Z^+$ for each element
$k \in X \cup Y$, and a target vector $< t_1,t_2,...,t_m >$ with positive integer entries
where $\sum_{i=1}^m t_i = \sum_{i=1}^m s(x_i)+\sum_{i=1}^m s(y_i)$. \\
{\bf Question}: Can $ X \cup Y$ be partitioned into $m$ disjoint sets $D_1,D_2,...,D_m$ each containing 
exactly one element from each of $X$ and $Y$ such that,
for $1 \leq i \leq m$, $\sum_{k \in D_i} s(k) = t_i$?

\begin{prop}\label{j2}
Problem $J2 | no-idle, no-wait | C_{\max}$ is $NP$-Hard in the strong sense.
\end{prop}
\begin{proof}
%We show that NMTS $\propto J2 | no-idle, no-wait | C_{\max}$.
We show that NMTS polynomially reduces to problem $J2 | no-idle, no-wait | C_{\max}$.
For any given instance of NMTS, we generate in $O(m)$ time an instance
of $J2 | no-idle, no-wait | C_{\max}$ with $n=3m$ jobs in the following way.
Jobs $J_1,...,J_m$ and jobs $J_{2m+1},...,J_{3m}$ follow the 
$M_1 \rightarrow M_2$  processing route.  
Jobs $J_{m+1},...,J_{2m}$ follow the 
$M_2 \rightarrow M_1$  processing route. 
Let us introduce $P = \sum_{i=1}^m t_i + \sum_{i=1}^m s(x_i) + \sum_{i=1}^m s(y_i)$.
The processing times are as follows:
$p_{1,1},...,p_{1,2m} = 1$; 
$p_{1,2m+i} = t_i+3P, \; \forall i = 1,...,m$;
$p_{2,i} = s(x_i) +P, \; \forall i = 1,...,m$; $p_{2,m+i} = s(y_i) +2P, \; \forall i = 1,...,m$;
$p_{2,2m+1},...,p_{2,3m} = 2$. 

\bigskip

\noindent
From now on, we will denote jobs $J_1,...,J_m$ as the $x$-jobs (as they relate to set $X$ of NMTS).
Similarly, jobs $J_{m+1},...,J_{2m}$ are denoted as the $y$-jobs and
jobs $J_{2m+1},...,J_{3m}$ are denoted as the $t$-jobs. 
Also, we will refer to a pattern $\alpha-\beta-\gamma$ whenever a triplet $\alpha-\beta-\gamma$ is consecutively
repeated on one machine. As an example, the triplet $x-t-y$ on $M_1$ indicates that a $x$-job is immediately followed by a $t$-job and then by a $y$-job on $M_1$. Correspondingly, the pattern $x-t-y$ on $M_1$ indicates that the sequence on $M_1$ is given
by $m$ consecutive triplets $x-t-y$.

\bigskip

\noindent
The total load of machine $M_1$ is given by $L_1 = \sum_{i=1}^{3m} p_{1,i} = 2m + 3mP + \sum_{i=1}^{m} t_i$.
Also,  the total load of machine $M_2$ is given by 
$L_2 = \sum_{i=1}^{3m} p_{2,i} = 2m + 3mP + \sum_{i=1}^{m} s(x_i)+ \sum_{i=1}^{m} s(y_i)$.
Hence $L_1 = L_2$.
Then, $L = \max \{L_1,L_2\} = L_1 = L_2$ constitutes a trivial lower bound on the makespan of the problem.
We remark that, on the one hand, the processing times on $M_1$ are unitary for $x$-jobs and $y$-jobs and have value $> 3P$ 
for $t$-jobs. On the other hand, the processing times on $M_2$ have value $2$ for $t$-jobs and value $> P$ for $x$-jobs and 
$> 2P$ for $y$-jobs respectively.
\bigskip

\noindent
In any feasible no-idle no-wait schedule, 
on $M_1$  every $t$-job (except possibly the last one) must be immediately followed by a $y$-job first and then by a $x$-job,
that is a $t$-job necessarily induces on $M_1$ a triplet $t-y-x$.
Indeed, two $t$-jobs cannot be adjacent on $M_1$, or else there would be idle time on $M_2$ between the completion of the 
first $t$-job and the start of the second $t$-job.
Also, a $t$-job cannot be immediately followed by an $x$-job on $M_1$, or else
the $x$-job would be waiting one unit of time between its completion on $M_1$ and its start on $M_2$.
Finally, a triplet $t-y-y$ cannot hold as it would induce
an idle time between the two $y$-jobs on $M_1$ 
and a triplet $t-y-t$ cannot hold as it would induce an idle time between the two $t$-jobs on $M_2$.
As a consequence, as there are $m$ $t$-jobs, a feasible schedule on $M_1$ follows either 
pattern $t-y-x$, or pattern $y-x-t$ or pattern $x-t-y$. 

\bigskip
\noindent
First consider pattern $y-x-t$ on $M_1$. This means that the first $y$-job, say job $J_j$, would be the first job in the sequence of $M_1$
and (due to the $M_2 \rightarrow M_1$ processing order) would be processed first
on machine $M_2$ so that $C_{1,j}=C_{2,j} + 1$. Hence, the following job on $M_2$ is necessarily the $x$-job, say job $J_i$, with $C_{1,i}=C_{1,j}+1=C_{2,j}+2$ which means that there is an idle time on $M_2$ between the end of job $J_j$ and the beginning of job $J_i$. So, pattern 
$y-x-t$ on $M_1$ cannot lead to a feasible schedule.

%\noindent
%For any given triplets, let denote by $J_i$ the related $x$-job,
%by $J_j$ the related $y$-job and by $J_k$ the related $t$-job.
%Thus, in order to stick to the no-idle no-wait constraint, on $M_2$, we must have
%$C_{2,j} = C_{1,k}$, $C_{2,k} = C_{2,j}+2 = C_{1,i}$ and $C_{2,i} =C_{1,i}+p_{2,i}$.
%This corresponds to have on $M_2$ a triplet $y-t-x$.
%Figure \ref{fig_t_y_x} shows the typical shape of a feasible schedule. 

%\begin{figure*}[!ht]
%    \begin{center}
%        \includegraphics[height=8.5cm,clip=true]{picture_t_y_x.pdf}
%\vspace*{-5.2cm}				
%        \caption[Fig]{Shape of a feasible schedule for problem $J2 | no-idle, no-wait | C_{\max}$}
%        \label{fig_t_y_x}
%    \end{center}
%\end{figure*}

%\begin{figure*}[!ht]
%    \begin{center}
%        \input{figure9}			
%        \caption[Fig]{Shape of a feasible schedule for problem $J2 | no-idle, no-wait | C_{\max}$}
%        \label{fig_t_y_x}
%    \end{center}
%\end{figure*}

\bigskip

\noindent
Hence, every feasible sequence on $M_1$ either follows pattern 
$x-t-y$ or pattern $t-y-x$. 

\bigskip

\noindent
In the first case,
the pattern on $M_2$ is $x-y-t$ and the relevant Gantt diagram is shown in Figure 
\ref{fig_x_y_t} for the case of $6$ jobs.
In this case, the makespan is given by the load of $M_2$ plus the processing time on $M_1$ of the first $x$-job,
that is $C_{\max} = L_2 +1 = L+1$.
Further, note that each
$t$-job $J_k$ on $M_1$ starts and completes when an $x$-job $J_j$ starts and a $y$-job $J_i$ completes on $M_2$, respectively.
Also, the sum $p_{2,i}+p_{2,j}$ of the processing times on $M_2$ of 
the $x$-job and the $y$-job is equal to the processing time $p_{1,k}$ on $M_1$ of the $t$-job.
But then, it provides a true assignment to the corresponding NMTS problem,
as $p_{2,i}+p_{2,j} = s(x_i)+s(y_j) +3P$ and $p_{1,k}=t_k+3P$, that is $s(x_i)+s(y_j) = t_k$.

%\begin{figure*}[!ht]
%    \begin{center}
%        \includegraphics[height=8.5cm,clip=true]{picture_x_t_y.pdf}
%\vspace*{-5.2cm}				
%        \caption[Fig]{The $x-t-y$ ($M_1$) / $x-y-t$ ($M_2$) patterns sequence}
%        \label{fig_x_y_t}
%    \end{center}
%\end{figure*}

\begin{figure*}[!ht]
    \begin{center}
        \ifx\JPicScale\undefined\def\JPicScale{1}\fi
\unitlength \JPicScale mm
\begin{picture}(110,19)(0,0)
\put(0,13){\makebox(0,0)[cc]{$M_1$}}

\put(1,3){\makebox(0,0)[cc]{$M_2$}}

\linethickness{0.3mm}
\put(6,10){\line(1,0){104}}
\linethickness{0.3mm}
\put(6,0){\line(1,0){104}}
\linethickness{0.3mm}
\put(6,16){\line(1,0){6}}
\put(6,10){\line(0,1){6}}
\put(12,10){\line(0,1){6}}
\put(6,10){\line(1,0){6}}
\put(8,19){\makebox(0,0)[cc]{$x$-job}}

\linethickness{0.3mm}
\put(12,10){\line(1,0){33}}
\put(12,10){\line(0,1){6}}
\put(45,10){\line(0,1){6}}
\put(12,16){\line(1,0){33}}
\put(28,13){\makebox(0,0)[cc]{$t$-job}}

\linethickness{0.3mm}
\put(12,6){\line(1,0){14}}
\put(12,0){\line(0,1){6}}
\put(26,0){\line(0,1){6}}
\put(12,0){\line(1,0){14}}
\linethickness{0.3mm}
\put(26,6){\line(1,0){19}}
\put(26,0){\line(0,1){6}}
\put(45,0){\line(0,1){6}}
\put(26,0){\line(1,0){19}}
\put(19,3){\makebox(0,0)[cc]{$x$-job}}

\put(36,3){\makebox(0,0)[cc]{$y$-job}}

\put(45,19){\makebox(0,0)[cc]{$y$-job}}

\linethickness{0.3mm}
\put(45,0){\line(1,0){12}}
\put(45,0){\line(0,1){6}}
\put(57,0){\line(0,1){6}}
\put(45,6){\line(1,0){12}}
\put(50,3){\makebox(0,0)[cc]{$t$-job}}

\linethickness{0.3mm}
\put(45,16){\line(1,0){6}}
\put(45,10){\line(0,1){6}}
\put(51,10){\line(0,1){6}}
\put(45,10){\line(1,0){6}}
\linethickness{0.3mm}
\put(51,16){\line(1,0){6}}
\put(51,10){\line(0,1){6}}
\put(57,10){\line(0,1){6}}
\put(51,10){\line(1,0){6}}
\put(56,19){\makebox(0,0)[cc]{$x$-job}}

\linethickness{0.3mm}
\put(57,10){\line(1,0){33}}
\put(57,10){\line(0,1){6}}
\put(90,10){\line(0,1){6}}
\put(57,16){\line(1,0){33}}
\put(73,13){\makebox(0,0)[cc]{$t$-job}}

\linethickness{0.3mm}
\put(57,6){\line(1,0){20}}
\put(57,0){\line(0,1){6}}
\put(77,0){\line(0,1){6}}
\put(57,0){\line(1,0){20}}
\linethickness{0.3mm}
\put(77,6){\line(1,0){13}}
\put(77,0){\line(0,1){6}}
\put(90,0){\line(0,1){6}}
\put(77,0){\line(1,0){13}}
\put(68,3){\makebox(0,0)[cc]{$x$-job}}

\put(84,3){\makebox(0,0)[cc]{$y$-job}}

\put(92,19){\makebox(0,0)[cc]{$y$-job}}

\linethickness{0.3mm}
\put(90,0){\line(1,0){12}}
\put(90,0){\line(0,1){6}}
\put(102,0){\line(0,1){6}}
\put(90,6){\line(1,0){12}}
\put(95,3){\makebox(0,0)[cc]{$t$-job}}

\linethickness{0.3mm}
\put(90,16){\line(1,0){6}}
\put(90,10){\line(0,1){6}}
\put(96,10){\line(0,1){6}}
\put(90,10){\line(1,0){6}}
\end{picture}			
        \caption[Fig]{The $x-t-y$ on $M_1$ / $x-y-t$ on $M_2$ patterns sequence}
        \label{fig_x_y_t}
    \end{center}
\end{figure*}
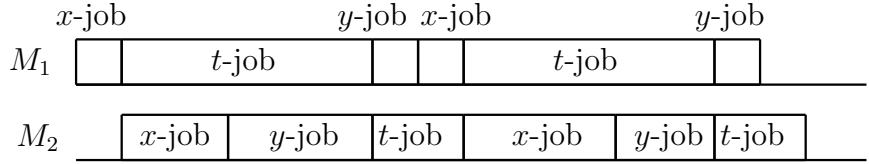

\noindent
In the latter case, the sequence on $M_1$ follows pattern $t-y-x$, that is starts with a $t$-job and 
ends with an $x$-job. The corresponding pattern on $M_2$ is $y-t-x$.
Hence, as the $x$-jobs follow the 
$M_1 \rightarrow M_2$  processing route, the last $x$-job on $M_1$ will then need to be processed also on $M_2$.
Thus, if we denote by $J_h$ this last $x$-job, the makespan will be $C_{\max} = L_1 + p_{2,h} > L_1 + P = L + P$ which is worse than the schedule obtained with pattern $x-t-y$ on $M_1$.\\

To conclude, if there exists a feasible schedule for the two-machine job shop problem with the shape $x-t-y$ on $M_1$ and $x-y-t$ on $M_2$
and value $=L+1$,
then this schedule is optimal and enables to derive in $O(m)$ time a yes answer to NMTS. Conversely, if there is no feasible solution to the job shop problem or the value is $> L+P$, then the answer to NMTS is no. This shows that the two-machine job shop problem is NP-hard. 
%The NP-hardness in the strong sense comes from the fact that the %proposed reduction is also a pseudo-polynomial reduction (see %\cite{GJ82} for the definition of a pseudo-polynomial reduction).

%\bigskip

%\noindent
%To summarize, either, the solution of the generated instance of the 
%$J2 | no-idle, no-wait | C_{\max}$ problem exists and has value $L+1$ 
%and correspondingly provides in $O(m)$ time a true assignment to the related NMTS problem,
%or there is no feasible solution with value $L+1$ (either no feasible schedule or 
%$C_{\max} > L+P$) and there exists no true assignment to the related NMTS problem.
\end{proof}

\section{Two-machine no-idle no-wait open shop scheduling}
\label{jss}
In this section we consider problem $O2 | no-idle, no-wait | C_{\max}$.
In the open shop configuration, 
no specific routing is assumed for each
job which can be either started first on $M_1$ and then on $M_2$ or viceversa.
We establish the complexity of this problem 
%$J2 | no-idle, no-wait | C_{\max}$ 
%and 
%$O2 | no-idle, no-wait | C_{\max}$ 
%by proving their 
%$NP$-Hardness in the strong sense. 
by proving that it is
$NP$-Hard in the strong sense. 

\begin{prop}\label{o2}
Problem $O2 | no-idle, no-wait | C_{\max}$ is $NP$-Hard in the strong sense.
\end{prop}
\begin{proof}
We show that NMTS polynomially reduces to problem $O2 | no-idle, no-wait | C_{\max}$
in a way similar to Proposition \ref{j2}.
From any given instance of NMTS, we generate in $O(m)$ time an instance
of $O2 | no-idle, no-wait | C_{\max}$ with $3m$ jobs having the same processing times 
as in the job shop case.

%\bigskip
%\noindent
%We remark, that, clearly, also here, if there exists 
%a true assignment to the related NMTS problem, then a feasible pattern $x-t-y$ on $M_1$ (and corresponding
%pattern $x-y-t$ on $M_2$) exists where all the $t$-jobs and $x$-jobs follow the $M_1 \rightarrow M_2$  processing route
%and all the $y$-jobs follow the $M_2 \rightarrow M_1$ processing route, namely the same solution of the job shop problem.
%Note that, if a true assignment to the related NMTS problem, then
%also a specular feasible solution exists having pattern $y-t-x$ on $M_1$ and $y-x-t$ on $M_2$
%where all the $t$-jobs and $y$-jobs follow the $M_1 \rightarrow M_2$  processing route
%and all the $x$-jobs follow the $M_2 \rightarrow M_1$ processing route.
%The makespan is always given by $L+1$.

\bigskip
\noindent
%We remark that,
%if a true assignment to the related NMTS problem exists, no feasible schedule
%can have two consecutive $x$-jobs or two consecutive $y$-jobs on $M_2$,
%as in the first case the sum of their processing time is $< 3P$, while in the latter case
%is $> 4P$ and all the $t$-jobs have processing time on $M_1$ superior to $3P$ and inferior to $4P$.
To show the reduction, we first show that no feasible solution exists with $C_{\max}=L$
and that, if a feasible solution exists with value $C_{\max}=L+1$, then
a true assignment to the related NMTS problem holds and can be derived in $O(m)$ time from the solution
of the $O2 | no-idle, no-wait | C_{\max}$  problem.

\bigskip
\noindent
In order to have $C_{\max}=L$, both machines should start processing at time $0$.
However, if an $x$-job (or a $y$-job) starts at time $0$ on $M_1$, then
no job can start at time $0$ on $M_2$ or else the no-wait requirement on the $x$-job
(the $y$-job) would be violated as all processing times on $M_2$ have length $\geq 2$.
Besides,  if a $t$-job starts at time $0$ on $M_1$, whatever job 
starting at time $0$ on $M_2$ would then violate the no-wait requirement 
as any processing time on $M_2$ has length inferior to the processing time
of any $t$-job on $M_1$. To fulfil the no-wait requirement, the considered job on $M_2$
would need to start at a time $t_0 > > 0$ but then $C_{\max}$ would be much larger than $L$. 
Thus, no feasible solution exists with $C_{\max}=L$.

\bigskip
\noindent
In order to have $C_{\max}=L+1$, the first job cannot start on any machine after time $1$.
Notice, that the same argument explicited above rules out the possibility of
having a $t$-job starting at time $t_0 \leq 1$ on $M_1$.
Also, if an $x$-job (a $y$-job) starts at time $1$  and completes at time $2$ on $M_1$,
in order to fulfil the no-wait requirement on the $x$-job
(the $y$-job), a $t$-job should start at time $0$ on $M_2$.
But then, the $t$-job will start at time $2$ on $M_1$ and the $x$-job ($y$-job) will start at time $2$ on $M_2$.
Correspondingly, to fulfil the no-idle no-wait requirement only a $y$-job ($x$-job) can be processed on $M_2$
once the $x$-job ($y$-job) is completed and the sum of the processing times of the $x$-job and $y$-job on $M_2$
must be equal to the processing time of the $t$-job on $M_1$.
This corresponds to the start of a pattern $x-t-y$ ($y-t-x$) on $M_1$ and $t-x-y$ ($t-y-x$) on $M_2$
where all the $t$-jobs and $y$-jobs ($x$- jobs) follow the $M_2 \rightarrow M_1$  processing route
and all the $x$-jobs ($y$-jobs) follow the $M_1 \rightarrow M_2$ processing route.
The resulting makespan is $L+1$. 
The corresponding two feasible schedules for a 6-job problem are depicted in Figure \ref{f1}.

%\begin{figure*}[!ht]
%    \begin{center}
%        \includegraphics[height=8.5cm,clip=true]{picture_f1.pdf}
%\vspace*{-5.2cm}				
%        \caption[Fig]{Two open shop patterns}
%        \label{f1}
%    \end{center}
%\end{figure*}

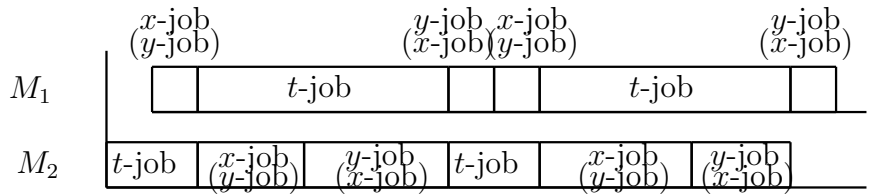
\begin{figure*}[!ht]
    \begin{center}
        \ifx\JPicScale\undefined\def\JPicScale{1}\fi
\unitlength \JPicScale mm
\begin{picture}(110,22)(0,0)
\put(0,13){\makebox(0,0)[cc]{$M_1$}}

\put(1,3){\makebox(0,0)[cc]{$M_2$}}

\linethickness{0.3mm}
\put(16,10){\line(1,0){94}}
\linethickness{0.3mm}
\put(10,0){\line(1,0){100}}
\linethickness{0.3mm}
\put(16,16){\line(1,0){6}}
\put(16,10){\line(0,1){6}}
\put(22,10){\line(0,1){6}}
\put(16,10){\line(1,0){6}}
\put(19,22){\makebox(0,0)[cc]{$x$-job}}

\linethickness{0.3mm}
\put(22,10){\line(1,0){33}}
\put(22,10){\line(0,1){6}}
\put(55,10){\line(0,1){6}}
\put(22,16){\line(1,0){33}}
\put(38,13){\makebox(0,0)[cc]{$t$-job}}

\linethickness{0.3mm}
\put(22,6){\line(1,0){14}}
\put(22,0){\line(0,1){6}}
\put(36,0){\line(0,1){6}}
\put(22,0){\line(1,0){14}}
\linethickness{0.3mm}
\put(36,6){\line(1,0){19}}
\put(36,0){\line(0,1){6}}
\put(55,0){\line(0,1){6}}
\put(36,0){\line(1,0){19}}
\put(29.38,3.75){\makebox(0,0)[cc]{$x$-job}}

\put(46,4){\makebox(0,0)[cc]{$y$-job}}

\put(55,22){\makebox(0,0)[cc]{$y$-job}}

\linethickness{0.3mm}
\put(55,0){\line(1,0){12}}
\put(55,0){\line(0,1){6}}
\put(67,0){\line(0,1){6}}
\put(55,6){\line(1,0){12}}
\put(60,3){\makebox(0,0)[cc]{$t$-job}}

\linethickness{0.3mm}
\put(55,16){\line(1,0){6}}
\put(55,10){\line(0,1){6}}
\put(61,10){\line(0,1){6}}
\put(55,10){\line(1,0){6}}
\linethickness{0.3mm}
\put(61,16){\line(1,0){6}}
\put(61,10){\line(0,1){6}}
\put(67,10){\line(0,1){6}}
\put(61,10){\line(1,0){6}}
\put(66,22){\makebox(0,0)[cc]{$x$-job}}

\linethickness{0.3mm}
\put(67,10){\line(1,0){33}}
\put(67,10){\line(0,1){6}}
\put(100,10){\line(0,1){6}}
\put(67,16){\line(1,0){33}}
\put(83,13){\makebox(0,0)[cc]{$t$-job}}

\linethickness{0.3mm}
\put(67,6){\line(1,0){20}}
\put(67,0){\line(0,1){6}}
\put(87,0){\line(0,1){6}}
\put(67,0){\line(1,0){20}}
\linethickness{0.3mm}
\put(87,6){\line(1,0){13}}
\put(87,0){\line(0,1){6}}
\put(100,0){\line(0,1){6}}
\put(87,0){\line(1,0){13}}
\put(78,4){\makebox(0,0)[cc]{$x$-job}}

\put(94,4){\makebox(0,0)[cc]{$y$-job}}

\put(102,22){\makebox(0,0)[cc]{$y$-job}}

\linethickness{0.3mm}
\put(100,16){\line(1,0){6}}
\put(100,10){\line(0,1){6}}
\put(106,10){\line(0,1){6}}
\put(100,10){\line(1,0){6}}
\put(19,19){\makebox(0,0)[cc]{($y$-job)}}

\put(55,19){\makebox(0,0)[cc]{($x$-job)}}

\put(66,19){\makebox(0,0)[cc]{($y$-job)}}

\put(102,19){\makebox(0,0)[cc]{($x$-job)}}

\linethickness{0.3mm}
\put(10,0){\line(1,0){12}}
\put(10,0){\line(0,1){6}}
\put(22,0){\line(0,1){6}}
\put(10,6){\line(1,0){12}}
\put(15,3){\makebox(0,0)[cc]{$t$-job}}

\put(29.38,1.25){\makebox(0,0)[cc]{($y$-job)}}

\put(46.25,1.25){\makebox(0,0)[cc]{($x$-job)}}

\put(78.12,1.25){\makebox(0,0)[cc]{($y$-job)}}

\put(94.38,1.25){\makebox(0,0)[cc]{($x$-job)}}

\linethickness{0.3mm}
\put(10,0){\line(0,1){18.12}}
\end{picture}				
        \caption[Fig]{Two open shop patterns}
        \label{f1}
    \end{center}
\end{figure*}

\bigskip
\noindent
Alternatively an $x$-job ($y$-job) must start at time $0$ on $M_1$ and at time $1$ on $M_2$ and, by applying the same reasoning
it turns out that the only way to have a feasible schedule with $C_{\max}=L+1$ is to stick to  
patterns $x-t-y$ ($y-t-x$) on $M_1$ and $x-y-t$ ($y-x-t$) on $M_2$. We remark that 
patterns $x-t-y$ on $M_1$ and $x-y-t$ on $M_2$ are the same patterns of the job shop problem,
hence Figure \ref{fig_x_y_t} for a $6$-job instance holds also here. Besides, by swapping every $x$-job
with a $y$-job in Figure \ref{fig_x_y_t}, we get the graphical representation of the solution
with patterns $y-t-x$ on $M_1$ and $y-x-t$ on $M_2$.

\bigskip
\noindent
By the same reasoning of the job shop case, we can see for all four cases, that if a 
feasible schedule with $C_{\max} = L+1$ exists, then it provides also a true assignment to the corresponding NMTS problem
by matching every pair of $x$-job, $y$-job to the related $t$-job.
Hence, given the solution of the generated instance of the $O2 | no-idle, no-wait | C_{\max}$ problem,
either a feasible solution exists with $C_{\max}=L+1$ and a true assignment to the corresponding NMTS problem
is then obtained in $O(m)$ time, or else either $C_{\max} > L+1$ or no feasible solution exists for 
problem $O2 | no-idle, no-wait | C_{\max}$ and correspondingly no true assignment 
exists for the NMTS problem. 
%As in the case of the job shop problem, this reduction is a pseudo-polynomial reduction thus proving the NP-hardness in the strong sense. 
\end{proof}

\section{Conclusions}

We considered no-idle/no-wait shop scheduling problems with makespan
as performance measure. We firstly focused on the $F2 | no-idle, no-wait | C_{\max}$
problem for which we established the connection to a special game of
dominoes, namely the Oriented Single Player Dominoes, and provided an $O(n)$ algorithm 
to solve both problems to optimality. As a byproduct, we also considered 
a special case of the Hamiltonian Path problem (denoted as 
Common/Distinct Successors Directed Hamiltonian Path) and proved that it 
reduces to problem $F2 | no-idle, no-wait | C_{\max}$. Correspondingly, we presented
a new polynomially solvable special case of the Hamiltonian Path problem.
Then, we extended our analysis to the more general $F2 | no-idle, no-wait | C_{\max}$
problem showing that it can be transformed into a two-machine
instance. Correspondingly, by applying the same approach as in the two-machine case,
we showed that it can be solved with complexity $O(mn \log n)$. For this problem we have also 
proposed a vectorial version of the Oriented Single Player Dominoes problem with tiles composed by vectors of numbers . Finally, we
proved that both the two-machine job shop and the two-machine open
shop problems are NP-hard in the strong sense by reduction from the
Numerical Matching with Target Sums.

\section*{Bibliography}
\label{sec:bib}

\end{document}